\documentclass[journal,twocolumn]{IEEEtran}

\usepackage{graphicx} 
\usepackage{bbm}
\usepackage{amsmath,amssymb,amsthm} 
\usepackage{chapterbib}
\usepackage{caption}
\usepackage{subcaption}
\usepackage{xcolor}
\usepackage{cite}
\usepackage{soul}

\usepackage{booktabs}
\usepackage{multirow}
\usepackage{graphicx}
\usepackage{amsmath,amssymb,amsthm} 
\usepackage[ruled,vlined]{algorithm2e}

\DeclareMathOperator*{\argmax}{argmax}
\DeclareMathOperator*{\Bin}{Bin}
\DeclareMathOperator*{\Bern}{Bern}
\newcommand{\bad}{{\mathcal B}}

\newcommand{\divergence}{\textup{Div}}

\def\Zsdp{\hat{Z}}
\def\SDP{semidefinite programming}

\def\twomodel{two-latent variable stochastic block model}
\def\Twomodel{Two-latent variable stochastic block model}
\def\twocensoredmodel{two-latent variable censored block model}
\def\Twocensoredmodel{Two-latent variable censored block model}
\makeatletter%
\if@twocolumn%
\newcommand{\Figwidth}{\columnwidth}%

\def\twocolbreak{\nonumber\\ &}%
\def\twocolAlignMarker{&}%
\def\onecolAlignMarker{}%
\else% \@twocolumnfalse
\newcommand{\Figwidth}{4.5in}%

\def\twocolbreak{}%
\def\twocolAlignMarker{}%
\def\onecolAlignMarker{&}%
\fi%
\makeatother%
\newtheorem{Theorem}{Theorem}
\newtheorem{Lemma}{Lemma}
\newtheorem{Remark}{Remark}
\newtheorem{Example}{Example}
\newtheorem{Corollary}{Corollary}

\definecolor{MyRed}{rgb}{139 0 0}
\definecolor{EndMyRed}{rgb}{0 0 0}

\def\qzero{q_0}
\def\qone{q_2}
\def\qtwo{q_1}
\def\qthree{q_3}

\def\Tone{T_1}
\def\Ttwo{T_2}
\def\Tthree{T_3}
\def\qbf{\mathbf{q}}
\def\gbf{\mathbf{g}}
\def\hbf{\mathbf{h}}
\def\Tset{\mathcal{T}}
\def\Hset{\mathcal{H}}

\def\Dist{P}
\def\Prob{{\mathbb P}}
\def\ahat{\hat{a}}
\def\bhat{\hat{b}}
\def\vprim{u}
\def\paaa{p}
\def\pbbb{\hat{p}}
\def\rhoy{\rho}

\newcommand{\lambdavec}[2]{{\lambda_{#1,#2}}}
\newcommand{\lambdavechat}[2]{{\hat{\lambda}_{#1,#2}}}
\newcommand{\lambdascalar}[4]{{\lambda^{(#3,#4)}_{#1,#2}}}
\newcommand{\lambdascalarhat}[4]{{\hat{\lambda}^{(#3,#4)}_{#1,#2}}}

\newcommand{\qij}[2]{{q^{(#1,#2)}}}
\newcommand{\gij}[2]{{g^{(#1,#2)}}}
\newcommand{\hij}[2]{{h^{(#1,#2)}}}

\newcommand{\dij}[2]{{d^{(#1,#2)}}}
\newcommand{\Dij}[2]{{D^{(#1,#2)}}}
\newcommand{\wij}[2]{{w^{(#1,#2)}}}
\newcommand{\Wij}[2]{{W^{(#1,#2)}}}

\newcommand{\qtildeij}[2]{{\Tilde{q}^{(#1,#2)}}}
\newcommand{\gtildeij}[2]{{\Tilde{g}^{(#1,#2)}}}
\newcommand{\htildeij}[2]{{\Tilde{h}^{(#1,#2)}}}

\title{Community Detection with Known, Unknown, or Partially Known Auxiliary Latent Variables}
\author{ Mohammad~Esmaeili and~Aria~Nosratinia,~\IEEEmembership{Fellow,~IEEE}
\thanks{M. Esmaeili and A. Nosratinia are with the Department of Electrical and Computer Engineering, The University of Texas at Dallas, Richardson, TX 75083-0688, USA, Email: esmaeili@utdallas.edu, aria@utdallas.edu. This work was supported in part by the NSF grant CIF-2008684.}
}

\begin{document}
\maketitle

\begin{abstract}
Empirical observations suggest that in practice, community membership does not completely explain the dependency between the edges of an observation graph. 
The residual dependence of the graph edges are modeled in this paper, to first order, by auxiliary node latent variables that affect the statistics of the graph edges but {\em carry no information about the communities of interest.} We then study community detection in graphs obeying the stochastic block model and censored block model with auxiliary latent variables. We analyze the conditions for exact recovery when these auxiliary latent variables are unknown, representing unknown nuisance parameters or model mismatch. We also analyze exact recovery when these secondary latent variables have been either fully or partially revealed. Finally, we propose a semidefinite programming  algorithm for recovering the desired labels when the secondary labels are either known or unknown. We show that exact recovery is possible by semidefinite programming  down to the respective maximum likelihood exact recovery threshold. 
\end{abstract}


%%%%%%%%%%%%%%%%%%%%%%%%%%%%%%%%%%%%%%%%%%%%%%%%%%%%%%%%%%%%%%%%%%%%%
\section{Introduction}
Community detection refers to a clustering of the nodes of a graph based on the observation of the edges. In many applications, this involves identifying groups of nodes that are more densely connected within the group than to nodes outside the group.
%Community detection, a form of unsupervised learning, refers to grouping (clustering) the nodes of a graph having similar affiliations (features).
Community detection has many applications such as finding like-minded people in social networks~\cite{Girvan2002}, exploration of biomedical networks~\cite{Chen2006}, improving link predictors and recommendation systems~\cite{berahmand2021preference, berahmand2021modified, Xu2014}, and is also relevant to network reconstruction problems~\cite{li2017reconstruction, huang2020sparse, peixoto2019network, wang2011network}. 
Community detection has been widely investigated in the literature from both theoretical and algorithmic perspectives. 
%The theoretical works on the community detection usually focus on synthetic generative models including stochastic block model, censored block model, and many others
Community detection is based on graph models such as the stochastic block model and the censored block model~\cite{holland1983stochastic, abbe2015community, hajek2015exact, esmaeili2020community, saade2015spectral, fronczak2013exponential, esmaeili2021community}. 
%Community detection utilizes several metrics for evaluating the residual error as the size of the graph grows. These metrics include 
Several metrics are used in this field to characterize the asymptotic behavior of the residual errors as the size of the graph grows, including correlated recovery, weak recovery, almost exact recovery, and exact recovery~\cite{decelle2011inference, mossel2015reconstruction, massoulie2014community, mossel2018proof, Esmaeili.BSSBM.Partially.Revealed, mossel2016density, yun2014community, abbe2017community, abbe2015exact, mossel2015consistency}.  
%Also, there are recovery techniques that are tractable theoretically including 
Among the various detection techniques one can name spectral methods, belief propagation, and semidefinite programming ~\cite{chen2016statistical, esmaeili2019community, mossel2014belief, Javad, amini2018semidefinite, hajek2016achieving}.

% Community detection refers to grouping (clustering) the nodes of a graph having similar affiliations (features). Community detection has been investigated widely in literature under various models~\cite{holland1983stochastic, abbe2015community, hajek2015exact, esmaeili2020community, saade2015spectral, fronczak2013exponential, esmaeili2021community}, recovery metrics~\cite{decelle2011inference, mossel2015reconstruction, massoulie2014community, mossel2018proof, Esmaeili.BSSBM.Partially.Revealed, mossel2016density, yun2014community, abbe2017community, abbe2015exact, mossel2015consistency}, and algorithms~\cite{chen2016statistical, esmaeili2019community, mossel2014belief, amini2018semidefinite, hajek2016achieving}. 

%Community detection is the act of estimating latent node variables (also called community labels) by observing the graph edges. 
In the graph models that have so far been studied for community detection, the graph edges are generated independently conditioned on the community labels. A brief survey of models that are most closely related to the present work will be presented shortly. However, in many practical community detection problems, the community labels do not fully explain the dependence between the graph edges. In other words, in many graphs encountered in practice, the graph edges conditioned on the desired community labels are not statistically independent. This happens when the structure of the graph is also influenced by factors other than the community of interest. For example, one may consider political affiliation communities on a social network in a university campus, where the social network graph is also influenced by other variables that may be unrelated to the community label of interest, such as membership in intramural and extramural activities. The nature and magnitude of the dependence of the graph on these secondary or auxiliary factors can have an effect on the performance of the community detection algorithm for the community label of interest. 
The present study  models and analyzes community detection in this scenario.

Toward that goal, this paper  introduces secondary or auxiliary latent variables in the graph model that are not subject to community detection themselves, but influence the structure of the graph. More specifically, we propose and employ a more general version of the stochastic block model and censored block model in which edges are independent conditioned on both the community labels and a set of secondary latent variables. The secondary or auxiliary latent variables represent a first-order model for the residual dependence of the edges of the graph once the effect of the community labels has been removed. Auxiliary variables are independent of community memberships and may or may not be observable. The auxiliary latent variable model is distinct from side-information model~\cite{esmaeili-journal,saad2018community} where the side information variables are directly observed and carry information about the communities. Side information represents non-graph information about communities, while auxiliary variables model the {\em graph connectivity patterns} that are {\em unrelated} to the communities.

We investigate the exact recovery threshold for community detection in the graphs with secondary latent variables. We also analyze the effect on the performance of community detection when this secondary latent variable is fully or partially known. We also propose and investigate a semidefinite programming  algorithm for community detection with secondary latent variables. Our analysis shows that exact recovery via semidefinite programming  is possible down to the respective maximum likelihood exact recovery threshold, for both unknown or known secondary latent variables.

In addition to addressing a novel problem, this paper also  provides a novel proof for bounding the summation of the minimums of Poisson-distributed values from above and below via Chernoff-Hellinger divergence. Our result (Lemma~\ref{Lemma-Poisson}) eliminates certain technical difficulties that existed in earlier proofs, e.g., does not impose restrictions on the domain of Poisson distributions. 
This result is extended (Lemma~\ref{Lemma-Poisson-two}) for the general censored block model. 
Also, the analysis of exact recovery for a graph generated based on two latent variables involves subtleties in extracting the maximum likelihood estimator and analyzing its semidefinite programming  relaxation, which go beyond earlier works.

To put the model of this paper in perspective, we review several community detection graph models whose nodes are associated, beyond a scalar community detection label, with some other variables too. The latent space model~\cite{hoff2002latent, ke2018information, sarkar2006dynamic} associates with each node a vector, often with small dimension, containing variables that are latent in the model. The graph edges are generated from a distribution that is parameterized based on the distance between the latent vectors of pairs of nodes, and the community is a scalar generated as a function of each latent vector. The  overlapping stochastic block model~\cite{abbe2015community,latouche2011overlapping} recovers multiple independent, identically distributed, binary communities via observing a graph whose edges are drawn independently conditioned on all the community labels of the terminating nodes. An important distinction of overlapped communities from the present work is that all communities must be recovered in the overlapped model, therefore the overlapped model has significant similarity with a multi-community model. In the overlapped model, the multiple communities posses a structure that can be exploited, compared with a general multi-community model. Finally, there exists some work on combining non-graph observation with graph observations~\cite{esmaeili-journal,saad2018community}; these works have a superficial resemblance to the subsection in this paper where the secondary latent variable is revealed. However, the graph and the side information in~\cite{esmaeili-journal,saad2018community} are assumed independent of each other conditioned on community labels, therefore the revealed side information in~\cite{esmaeili-journal,saad2018community} has no direct influence on the graph. Thus,~\cite{esmaeili-journal,saad2018community} model a different phenomenon and also have a different mathematical structure, compared with the present work.
In the interest of brevity, our coverage of various community detection models is limited, and the interested reader is referred to more comprehensive coverage available, e.g., in~\cite{abbe2015community}.

Notation: 
$\mathbf{I}$ is the identity matrix and $\mathbf{J}$ the all-one matrix. 
$S \succeq 0$ indicates a positive semidefinite matrix and  $S \ge 0$ denotes a matrix with non-negative entries. $||S||$ is the spectral norm and $\lambda_{2}(S)$ is the second smallest eigenvalue (for a symmetric matrix).
$[a,b]$ is a vector that is obtained by stacking vectors $a$ and $b$.  
$\langle \cdot,\cdot\rangle$ is the inner product and $*$ is the element-wise product. We abbreviate $[n] \triangleq \{ 1, \cdots,n \}$. 
% Probabilities are denoted by  $\Prob(\cdot)$.
$\Prob(\cdot)$  indicates the probability operator and $\Dist(\cdot)$ a probability distribution which is identified by the choice of its variables whenever there is no confusion.
Random variables with Bernoulli and Binomial distributions are indicated by $\Bern(p)$ and $\Bin(n,p)$, respectively, with $n$ trails and success probability $p$. Also, random variables with Poisson distribution are indicated by $\mathcal{P}_{\lambda}(n)$ with $n$ trails and parameter $\lambda$. 
%%%%%%%%%%%%%%%%%%%%%%%%%%%%%%%%%
\section{System Model}%`
\label{system-model}
We start by considering a two-latent variable model, and assume the cardinality of both is finite. For notational convenience throughout the paper, $x,y$ are length-$n$ vectors holding latent variable values for the whole graph, while the latent variables for any node $v$ are represented with $x_\nu, y_\nu$. In our model, we aim to discover $x$, therefore nodes that share the same value for $x$ are called a {\em community.} By {\em micro-community}, we refer to the set of nodes in the graph that share the same value for both latent variables $x,y$. The matrix $P$ denotes prior probabilities 
\[
P_{i,j} = \Prob(x_v=i,y_v=j).
\]
For convenience and for avoiding tensor calculations, we further define:
\[
p \triangleq \text{vec}(P).
\]

For both the \twomodel{} and \twocensoredmodel{}, the graph edges are Bernoulli distributed, conditioned on the latent variables of the two nodes terminating the edge.  The conditional Bernoulli parameters for an arbitrary edge are organized in a symmetric matrix $\bar{Q}$, whose rows and columns are ordered in a manner compatible with vector $p$. In other words, assuming the latent variable $x_v$ has $m_x$ outcomes, then the probability of an edge between two nodes with latent variable pairs taking values $(i,j)$ and $(i',j')$ is given by the element of $\bar{Q}$ in row $jm_x+i$ and column $j'm_x+i'$. 

We are interested in a regime where edge probabilities diminish with the size of the graph $n$, in particular, in the context of our model there exist a constant matrix $Q$ such that:
\begin{align*}
    \bar{Q} &= \frac{\log n}{n} Q .
\end{align*}
This assumption asymptotically guarantees a fully connected graph.
\begin{Example}
Consider a two-latent variable stochastic block model with $m_x = 2$ and $m_y = 3$. Then 
\begin{align*}
    P = \begin{bmatrix}
P_{0,0} & P_{0,1} & P_{0,2} \\ 
P_{1,0} & P_{1,1} & P_{1,2}
\end{bmatrix} ,
\end{align*}

\begin{align*}
p = 
\begin{bmatrix}
P_{0,0} & P_{0,1} & P_{0,2} & P_{1,0} & P_{1,1} & P_{1,2}
\end{bmatrix} , 
\end{align*}

\begin{align*}
\bar{Q} = \frac{\log n}{n}
    \begin{bmatrix}
Q_{0,0} & Q_{0,1} & Q_{0,2} & Q_{0,3} & Q_{0,4} & Q_{0,5} \\ 
Q_{1,0} & Q_{1,1} & Q_{1,2} & Q_{1,3} & Q_{1,4} & Q_{1,5} \\
Q_{2,0} & Q_{2,1} & Q_{2,2} & Q_{2,3} & Q_{2,4} & Q_{2,5} \\
Q_{3,0} & Q_{3,1} & Q_{3,2} & Q_{3,3} & Q_{3,4} & Q_{3,5} \\
Q_{4,0} & Q_{4,1} & Q_{4,2} & Q_{4,3} & Q_{4,4} & Q_{4,5} \\
Q_{5,0} & Q_{5,1} & Q_{5,2} & Q_{5,3} & Q_{5,4} & Q_{5,5}
\end{bmatrix}.
\end{align*}
\end{Example}
In addition, we define the columns of weighted versions of the matrix $Q$ as
\begin{align*}
    \qij{i}{j} & \triangleq \text{diag}(p) Q \,e_{jm_x+i} ~, 
\end{align*}
where $e_k$ is the $k$-th canonical coordinate vector, and for convenience our notation of $\qij{i}{j}$ emphasizes dependence on the latent variable outcomes rather than matrix coordinates.
Thus, $\qij{i}{j}$ is the column of $\text{diag}(p)Q$. This vector represents the relative frequency of edges connecting a node from the micro-community $(i,j)$ to all nodes of each micro-community (including the same micro-community).
Also, we define the vector $\qtildeij{i}{j}$ of size $m_x$ with entries
\begin{align*}
    \qtildeij{i}{j}_{i'} &\triangleq \sum_{j'}  P_{i',j'} Q_{j'm_x+i',jm_x+i} ~,
\end{align*}
representing the relative frequency of edges, connecting a node from the micro-community $(i,j)$ to all nodes of micro-communities with similar community latent variable.

For the \twocensoredmodel{}, if an edge exists between a pair of nodes, the sign of the edge (positive or negative) is determined by a random variable drawn from a Bernoulli distribution with a certain parameter.  
The Bernoulli parameters for  the positive sign of an edge are organized in a symmetric matrix $\Xi$, whose rows and columns are also ordered in a manner compatible with vector $p$.
Finally, for the censored block model, we define similarly
\begin{align*}
    \gij{i}{j} & \triangleq \text{diag}(p) (\Xi * Q) \,e_{jm_x+i} ~, \\
    \hij{i}{j} & \triangleq \text{diag}(p) ((1-\Xi) * Q) \,e_{jm_x+i} ~,
\end{align*}
and
\begin{align*}
    \gtildeij{i}{j}_{i'} &\triangleq \sum_{j'}  P_{i',j'} (\Xi * Q)_{j'm_x+i',jm_x+i} ~, \\
    \htildeij{i}{j}_{i'} &\triangleq \sum_{j'}  P_{i',j'} ((1-\Xi) * Q)_{j'm_x+i',jm_x+i} ~.
\end{align*}

\begin{Remark}
The censored block model in~\cite{hajek2016achievingExtensions, esmaeili2019community} with parameters $a$ and $\xi$ is a special case of the general censored model represented in this paper with
\begin{align*}
    Q  = \begin{bmatrix}
a & a\\ 
a & a
\end{bmatrix} , 
\quad \Xi  = \begin{bmatrix}
1-\xi & \xi\\ 
\xi & 1-\xi
\end{bmatrix}.
\end{align*}
\end{Remark}
%%%%%%%%%%%%%%%%%%%%%%%%%
\section{Exact Recovery under Optimal Detection}
\label{Information-Theoretic Results}
The main results of this part are represented in the context of three scenarios, where the latent variable $x$ is unknown and the latent variable $y$ is either known or unknown (for all nodes in the graph) or partially known (for some nodes in the graph).
Figure~\ref{fig: SBM-known} 
%provides some intuition about the problems that we consider in this paper. This figure 
shows graph realizations of a two-latent variable stochastic block model with $m_x=2$ and $m_y=2$. In each node, the community latent variable is indicated by the color of the inner circle, and the auxiliary latent  variable is represented by the color of a ring around the inner circle.

\begin{figure*}
\begin{center}
\begin{subfigure}{0.24\textwidth}
         \centering
         \includegraphics[width=\textwidth]{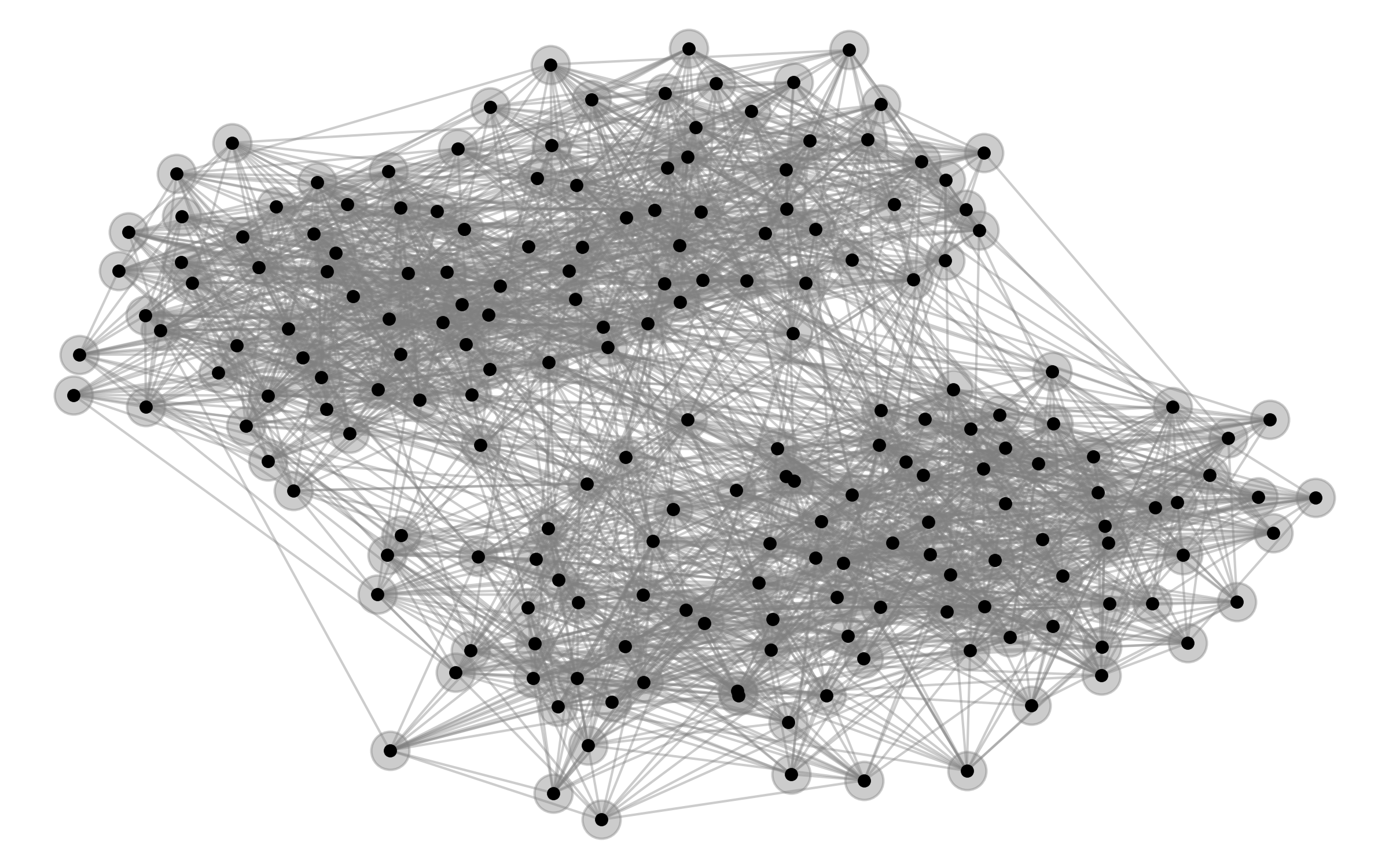}
         \caption{}
     \end{subfigure}
     \hfill
     \begin{subfigure}{0.24\textwidth}
         \centering
         \includegraphics[width=\textwidth]{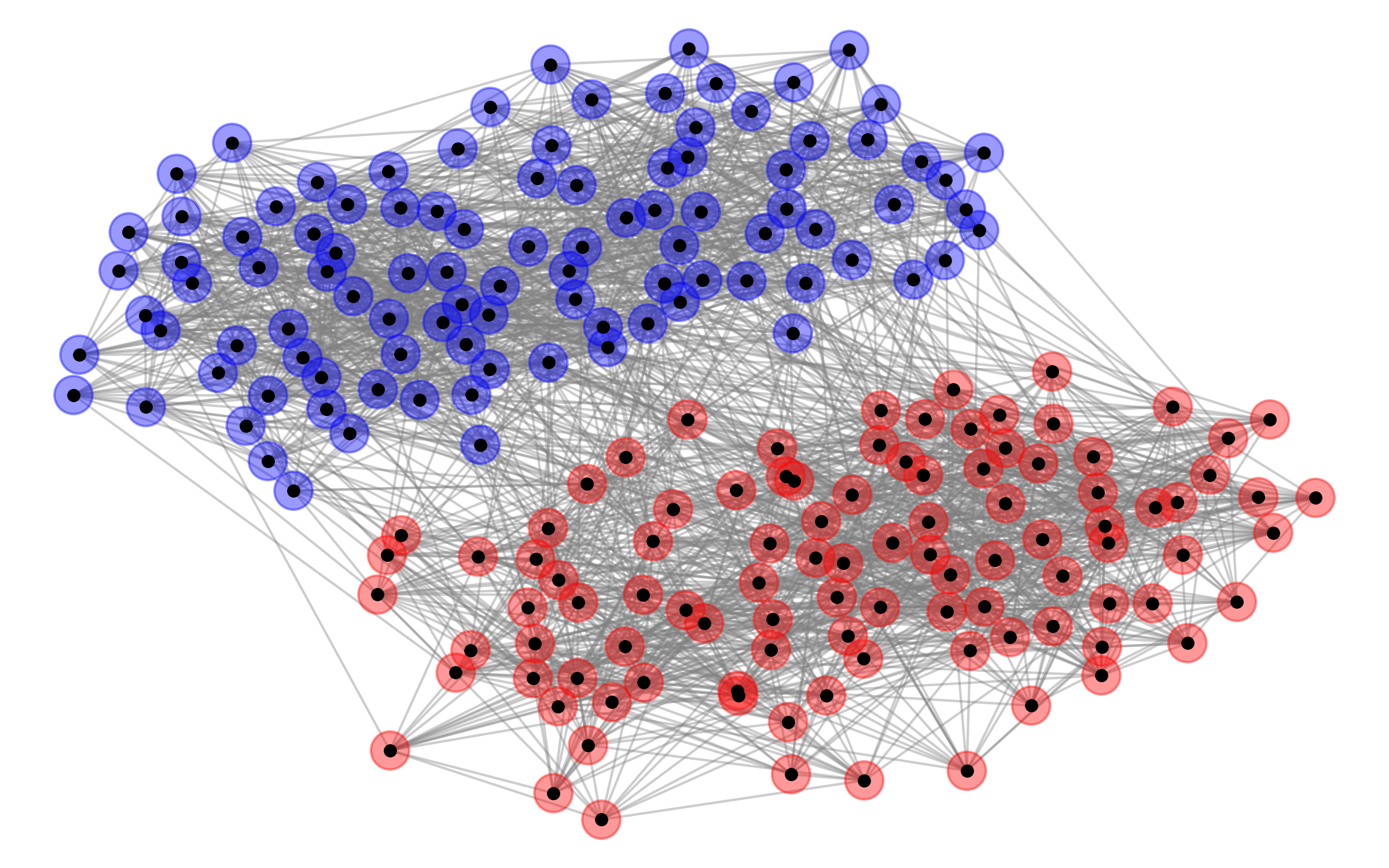}
         \caption{}
     \end{subfigure}
     \hfill
%\end{center}
% \caption{For each node, (a) both latent variables are unknown, (b) knowing the statistics of the graph, the community latent variable is recovered while the second one is unknown.}
% \label{fig: SBM-unknown}
% \end{figure*}
          \begin{subfigure}{0.24\textwidth}
         \centering
         \includegraphics[width=\textwidth]{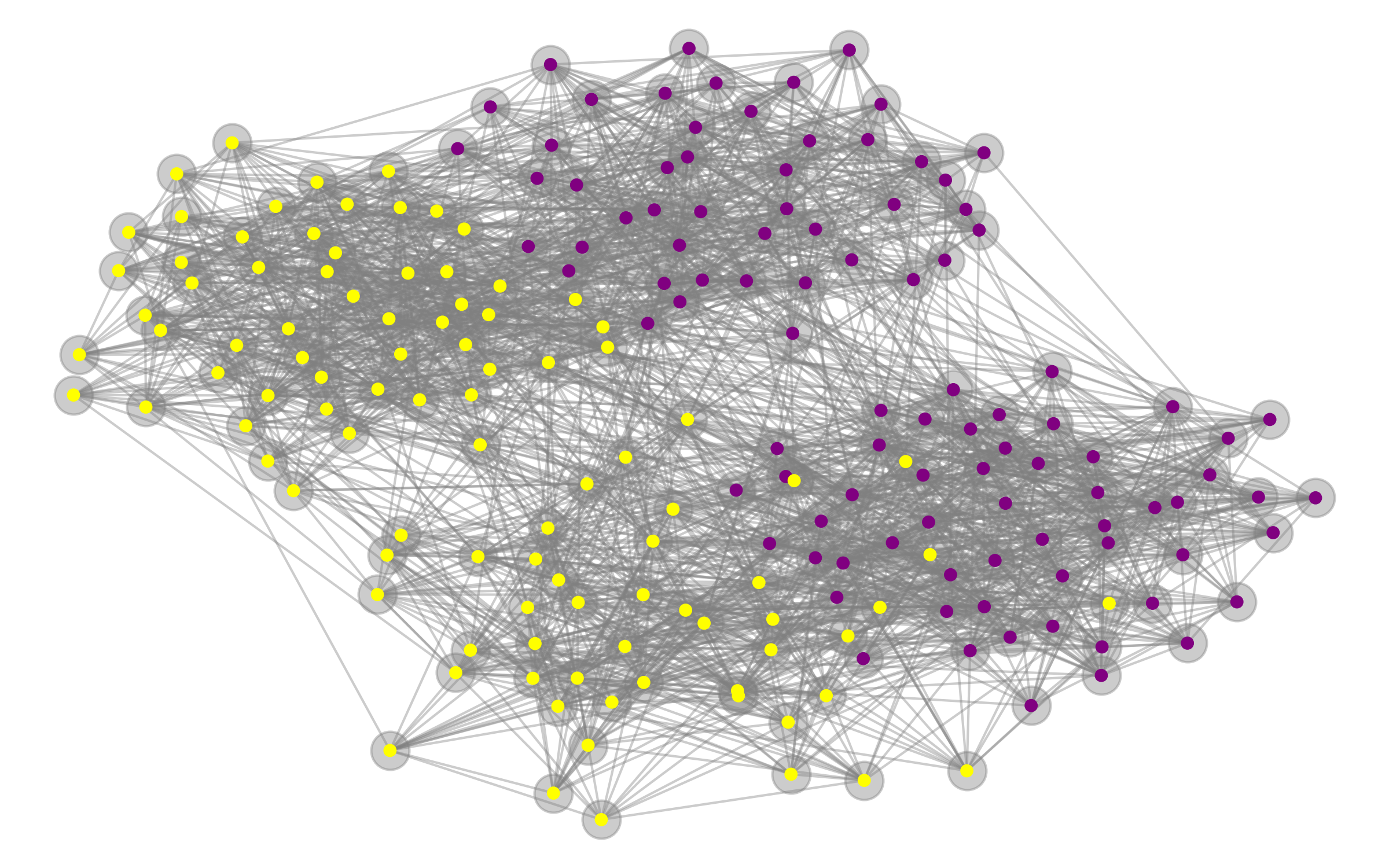}
         \caption{}
     \end{subfigure}
          \hfill
     \begin{subfigure}{0.24\textwidth}
         \centering
         \includegraphics[width=\textwidth]{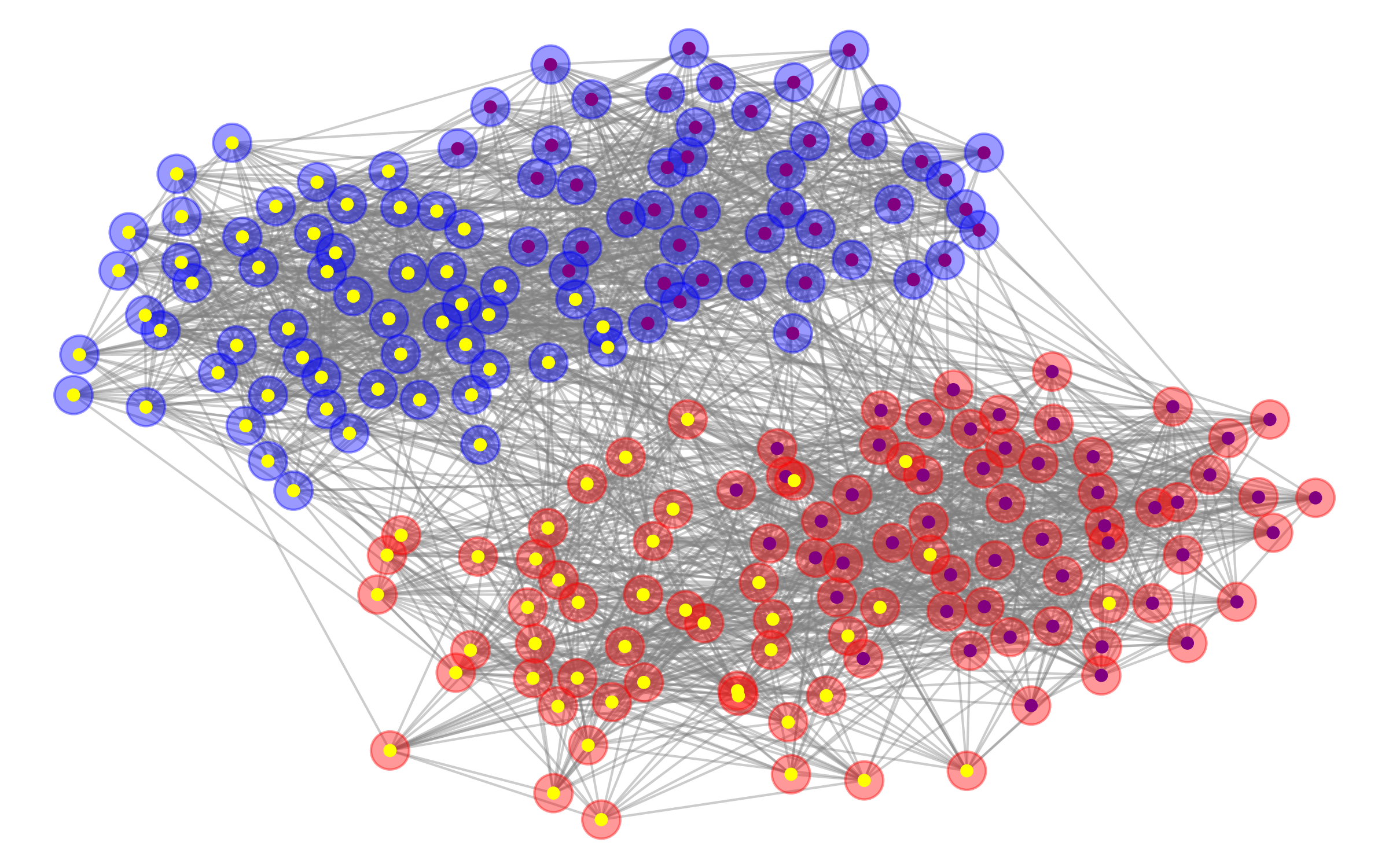}
         \caption{}
     \end{subfigure}
\end{center}
\caption{For each node, (a) both latent variables are unknown, (b) knowing the statistics of the graph, the community latent variable is recovered while the auxiliary latent variable is unknown, (c) the auxiliary latent  variable is known while the first one is unknown, (d) knowing the statistics of the graph, the community latent variable is recovered while the auxiliary latent variable is known.}
\label{fig: SBM-known}
\end{figure*}

The Chernoff-Hellinger divergence is due to Abbe~\cite{abbe2017community} and is defined for two non-negative vectors $a, b$ of the same dimension:
\begin{align}
\divergence(a,b) &\triangleq \max_{t \in [0, 1]} \sum_{i} \big[ta_{i}+(1-t)b_{i}-a_{i}^{t}b_{i}^{1-t} \big].
\label{eq:CH-definition}
\end{align}
This is a generalization of the Hellinger divergence and the Chernoff divergence~\cite{abbe2015community, abbe2017community}. In a manner similar to~\cite{abbe2015community} we present a lemma that bounds a summation of the minimums of Poisson-distributed values.

\begin{Lemma}
\label{Lemma-Poisson}
Let $a,b \in \mathbb{R}_{+}^{m}$, with $a\neq b$, and two positive scalars $\paaa,\pbbb$.
For any Poisson multivariate distributions $\mathcal{P}_{a}(d)$ and $\mathcal{P}_{b}(d)$, define
\begin{align*}
    I(a,b) \triangleq \sum_{d \in \mathbb{Z}_{+}^{m}} \min \{ \mathcal{P}_{a}(d) \paaa , \mathcal{P}_{b}(d) \pbbb \}. 
\end{align*}
Then
\begin{align*}
I(a,b) &\leq \max \{\paaa, \pbbb \} e^{-\divergence(a,b)} , \\
I(a,b) &\geq  \min \{\paaa, \pbbb \} e^{-\divergence(a,b)} \prod_{i=1}^{m} \frac{1}{e} \big(a_{i}^{t^{*}}b_{i}^{1-t^{*}}\big)^{-\frac{1}{2}},
\end{align*}
where $t^{*}$ is the optimal parameter in the definition of Chernoff-Hellinger divergence $\divergence(a,b)$.
\end{Lemma}
\begin{proof}
See Appendix~\ref{Proof-Lemma-Poisson}.
\end{proof}

Let $D$ be a random variable vector representing the number of edges that connect the node $v$ to each micro-community. More specifically, $\Dij{i'}{j'}$ is an element of the $D$ indicating the number of edges connecting the node $v$ to the micro-community $(i',j')$. 
For each node $v$, the proposed detection tests hypotheses
\begin{align*}
    H_i: x_v=i .
\end{align*}
If $v$ belongs to micro-community $(i,j)$, then 
\begin{align*}
    \Dij{i'}{j'} \sim \Bin(nP_{i',j'}, \bar{Q}_{j'm_x+i',jm_x+i}) .
\end{align*}
In the regime where $\bar{Q} = Q \frac{\log n}{n}$, the Binomial distribution can be approximated by a Poisson distribution with the same mean, denoted $\lambdascalar{i}{j}{i'}{j'}$. 
Indeed, using Le Cam's inequality, the total variation distance between $\Bin(nP_{i',j'}, \frac{\log n}{n} Q_{j'm_x+i',jm_x+i} )$ and $\mathcal{P}(P_{i',j'} Q_{j'm_x+i',jm_x+i} \log n)$ asymptotically goes to zero. 
Then
\begin{align*}
     \Prob (D = d|H_i,y_v=j) =\prod_{i'}\prod_{j'} \mathcal{P}_{\lambdascalar{i}{j}{i'}{j'}} (\dij{i'}{j'}) ,
\end{align*}
where $\lambdascalar{i}{j}{i'}{j'} = P_{i',j'} Q_{j'm_x+i',jm_x+i} \log n$.
\begin{Theorem}
Under the \twomodel{}, all micro-communities are exactly recovered if and only if 
\begin{align*}
\min_{(i,j) \neq (k,l)} \divergence(\qij{i}{j},\qij{k}{l}) > 1.
\end{align*}
\end{Theorem}
\begin{proof}
It follows from the exact recovery under the general stochastic bock model or the general overlapping stochastic block model. 
\end{proof}

\begin{Theorem}
\label{theorem-general-side}
Under the \twomodel{}, when the latent variable $y$ is revealed, exact recovery of $x$ is possible if and only if
\begin{align*}
\gamma_1 \triangleq \min_j \min_{i\neq k} \divergence(\qij{i}{j},\qij{k}{j}) > 1.
\end{align*}
\end{Theorem}
\begin{proof}
See Appendix~\ref{proof-theorem-general-side}.
\end{proof}

\begin{Theorem}
\label{theorem-general-no side}
Under the \twomodel{}, when both latent variables are unknown, exact recovery of $x$ is possible if and only if
\begin{align*}
\gamma_2 \triangleq  \min_{j} \min_{i\neq k} \divergence \big( \qtildeij{i}{j}, \qtildeij{k}{j} \big ) > 1.
\end{align*}
\end{Theorem}
\begin{proof}
See Appendix~\ref{proof-theorem-general-no side}.
\end{proof}

Now we present the following Lemma which is similar to Lemma~\ref{Lemma-Poisson} and is crucial for the analysis of the censored block model.  
\begin{Lemma}
\label{Lemma-Poisson-two}
Let $a,b,\ahat,\bhat \in \mathbb{R}_{+}^{m}$, with $a\neq b$ or $\ahat \neq \bhat$, and two positive scalars $\paaa,\pbbb$.
For any Poisson multivariate distributions $\mathcal{P}_{a}(d)$, $\mathcal{P}_{b}(d)$, $\mathcal{P}_{\ahat}(w)$, and $\mathcal{P}_{\bhat}(w)$, define
\begin{align*}
    I(a,b, \ahat, \bhat) \triangleq \sum_{d,w \in \mathbb{Z}_{+}^{m}} \min \{ \mathcal{P}_{a}(d) \mathcal{P}_{\ahat}(w) \paaa, \mathcal{P}_{b}(d) \mathcal{P}_{\bhat}(w) \pbbb \} . 
\end{align*}
Then
\begin{align*}
I(a,b, \ahat, \bhat) \leq& \max \{\paaa, \pbbb \} e^{-\divergence([a,\ahat],[b,\bhat])} , \\
I(a,b, \ahat, \bhat) \geq&  \min \{\paaa, \pbbb \}  e^{-\divergence([a,\ahat],[b,\bhat])} \\
&\times \prod_{i} \frac{1}{e^2} \big [ (a_{i} \ahat_{i} )^{t^{*}} (b_{i}\bhat_{i})^{1-t^{*}} \big ]^{-\frac{1}{2}} ,
\end{align*}
where $t^{*}$ is the optimal parameter in the definition of Chernoff-Hellinger divergence $\divergence([a,\ahat],[b,\bhat])$.
\end{Lemma}
\begin{proof}
See Appendix~\ref{Proof-Lemma-Poisson-two}.
\end{proof}

Let $D$ and $W$ be random vectors representing the positive and negative edges that connect the node $v$ to each micro-community, respectively. More specifically, $\Dij{i'}{j'}$ and $\Wij{i'}{j'}$ are elements of $D$ and $W$ indicating the number of positive and negative edges connecting the node $v$ to the micro-community $(i',j')$, respectively.
For each node $v$, the proposed detection tests hypotheses
\begin{align*}
    H_i: x_v=i .
\end{align*}
If $v$ belongs to micro-community $(i,j)$, then 
\begin{align*}
    &\Dij{i'}{j'} \sim \Bin(nP_{i',j'}, (\Xi * \bar{Q})_{j'm_x+i',jm_x+i}) ,\\
    &\Wij{i'}{j'} \sim \Bin(nP_{i',j'}, ((1-\Xi) * \bar{Q})_{j'm_x+i',jm_x+i}) .
\end{align*}
In the regime where $\bar{Q} = Q \frac{\log n}{n}$, the Binomial distribution can be approximated by a Poisson distribution with the same mean. 
The distributions of $D$ and $W$ can be approximated by multivariate Poisson distributions $\mathcal{P}_{\lambdavec{i}{j}}$ and $\mathcal{P}_{\lambdavechat{i}{j}}$
with the vector means $\lambdavec{i}{j}$ and $\lambdavechat{i}{j}$, respectively.
Therefore
\begin{align*}
\Prob (D = d &,W = w |H_i,y_v=j) \\
    &=  \Prob (D = d|H_i,y_v=j)  \Prob (W = w|H_i,y_v=j) \\
    &=\prod_{i'}\prod_{j'} \mathcal{P}_{\lambdascalar{i}{j}{i'}{j'}} (\dij{i'}{j'}) \mathcal{P}_{\lambdascalarhat{i}{j}{i'}{j'}} (\wij{i'}{j'}) ,
\end{align*}
where 
\begin{align*}
    \lambdascalar{i}{j}{i'}{j'} &= P_{i',j'} (\Xi *Q)_{j'm_x+i',jm_x+i} \log n ,\\
    \lambdascalarhat{i}{j}{i'}{j'} &= P_{i',j'} ((1-\Xi) * Q)_{j'm_x+i',jm_x+i} \log n .
\end{align*}

\begin{Theorem}
\label{theorem-general-cencored}
Under \twocensoredmodel{}, all micro-communities are exactly recovered if and only if 
\begin{align*}
\min_{(i,j) \neq (k,l)} \divergence \big ([\gij{i}{j}, \hij{i}{j}],[\gij{k}{l}, \hij{k}{l}] \big) > 1.
\end{align*}
\end{Theorem}
\begin{proof}
See Appendix~\ref{proof-theorem-general-cencored}.
\end{proof}

\begin{Theorem}
\label{theorem-general-cencored-side}
Under the \twocensoredmodel{}, when the latent variable $y$ is revealed, exact recovery of $x$ is possible if and only if
\begin{align*}
\gamma_3 \triangleq \min_j \min_{i\neq k} \divergence \big ([\gij{i}{j}, \hij{i}{j}] , [\gij{k}{j}, \hij{k}{j}] \big)> 1.
\end{align*}
\end{Theorem}
\begin{proof}
See Appendix~\ref{proof-theorem-general-cencored-side}.
\end{proof}

\begin{Theorem}
\label{theorem-general-cencored-no side}
Under the \twocensoredmodel{}, when both latent variables are unknown, exact recovery of $x$ is possible if and only if
\begin{align*}
\gamma_4 \triangleq  \min_{j} \min_{i\neq k} \divergence \big( [\gtildeij{i}{j}, \htildeij{i}{j}], [\gtildeij{k}{j}, \htildeij{k}{j}] \big ) > 1.
\end{align*}
\end{Theorem}
\begin{proof}
See Appendix~\ref{proof-theorem-general-cencored-no side}.
\end{proof}

\begin{Corollary}
Assume $x$ and $y$ are unknown latent variables for all nodes. We randomly reveal the latent variable $y$ for $(1-\epsilon)n$ nodes, where $\epsilon \in (0,1)$. This is equivalent to erasing the latent variable $y$ which is a known latent variable from a node with erasure probability $\epsilon$. Define 
\begin{align*}
\beta_{1} \triangleq - \lim_{n\rightarrow \infty} \frac{\log (1-\epsilon)}{\log n}, \quad 
\beta_{2} \triangleq - \lim_{n\rightarrow \infty} \frac{\log \epsilon}{\log n}.
\end{align*}
\begin{itemize}
    \item 
    Under the \twomodel{} exact recovery is asymptotically possible for latent variable $x$ if and only if
    \begin{align*}
        \min \big( \gamma_{1}+\beta_{1}, \gamma_{2}+\beta_{2}\big) > 1.
    \end{align*}
    \item 
    Under the \twocensoredmodel{} exact recovery is asymptotically possible for latent variable $x$ if and only if
    \begin{align*}
        \min \big( \gamma_{3}+\beta_{1}, \gamma_{4}+\beta_{2}\big) > 1.
    \end{align*}
\end{itemize}
\end{Corollary}
The results of this part generalize to $M$ latent variables without difficulty.

\begin{Remark}
To prove the “if” part of all theorems in Section~\ref{Information-Theoretic Results}, a partial recovery algorithm is required before applying a MAP estimator. 
For that purpose, the partial recovery algorithm in~\cite{abbe2015community} is adopted and modified to match the scenarios in this paper. Please see Appendix~\ref{Partial Recavery Algorithm}.
\end{Remark}
%%%%%%%%%%%%%%%%%%%%%%%%%%%%%%%%%%%%%%%%%%%%%
\section{Semidefinite Programming Results}
\label{Semidefinite Programming Results}
This section describes a semidefinite programming algorithm for recovering the desired latent variable. 
The main results of this part are represented in the context of two scenarios, where the latent variable $x$ is unknown and the latent variable $y$ is either known or unknown (for all nodes in the graph).
We consider $x, y \in \{\pm 1\}^{n}$ such that $x^{T}\mathbf{1} = 0$. Thus, the latent variable $x$ represents two equal-sized communities. The sample size of the latent variable $y$, represented by $\rhoy \triangleq  \frac{1}{n}|\{v \in [n]: y_{v} = 1\}|$, is an unknown quantity.\footnote{Note that semidefinite programming results in this section are obtained for binary equal-sized communities, while the results of Section~\ref{Information-Theoretic Results} were more general.}

\subsection{\Twomodel{}}
We highlight the specifics of a \twomodel{} for the purposes of upcoming calculations.
The probability of an edge drawn between two nodes $v, \vprim$ is characterized by four constants, $\qzero, \qtwo, \qone, \qthree$ such that:
\begin{align*}
A_{ij} \sim  \begin{cases}
\Bern( \qzero \frac{\log n}{n}) & \text{if} \quad x_{v}=x_{\vprim}, y_{v}=y_{\vprim}\\
\Bern( \qtwo \frac{\log n}{n}) & \text{if} \quad x_{v}\neq x_{\vprim}, y_{v}=y_{\vprim}\\
\Bern( \qone \frac{\log n}{n}) & \text{if} \quad x_{v}=x_{\vprim}, y_{v}\neq y_{\vprim}\\
\Bern( \qthree \frac{\log n}{n}) & \text{if} \quad x_{v}\neq x_{\vprim}, y_{v}\neq y_{\vprim}
\end{cases} .
\end{align*}
The corresponding matrix $Q$, as defined earlier, in this case will be:
\begin{align}
\label{equ:Q}
Q =\begin{bmatrix}
\qzero & \qtwo  & \qone  & \qthree\\ 
\qtwo  & \qzero & \qthree & \qone \\ 
\qone  & \qthree & \qzero & \qtwo \\ 
\qthree & \qone  & \qtwo  & \qzero
\end{bmatrix} .
\end{align}

\subsubsection{Recovering $x$ when $y$ is known}
\label{subsec:sbm-one unknown}
In the first scenario, given an observation of the graph $A$ and $y$ which corresponds to the observed graph, the latent variable $x_{v}$ is recovered exactly for each node $v \in [n]$. In this part, $y$ is considered as an observation which helps the estimator to recover the desired latent variable $x$.
Let $W \triangleq yy^{T}$ and $B \triangleq W * A$. Since $x$ is chosen uniformly over $\{x\in\{\pm 1\}^{n} : x^{T}\mathbf{1}=0\}$, the maximum likelihood estimator gives the optimal solution. For this configuration, the log-likelihood is
\begin{equation*}
\log \Prob (A|x, y) = \frac{\Tone}{8} x^{T}Bx +\frac{\Ttwo}{8} x^{T}Ax + c, 
\end{equation*}
where $\Tone  \triangleq \log  (\frac{\qzero \qthree}{\qone \qtwo } )$ and $\Ttwo \triangleq  \log  (\frac{\qzero \qone }{\qtwo \qthree} )$, as $n \rightarrow \infty$ and $c$ is a constant.
Considering the constraints, the maximum likelihood estimator is,
\begin{align}
\label{oublv-ml-equ1}
\hat{x} =&  \underset{x}{\arg\max} ~\Tone x^{T}Bx +\Ttwo x^{T}Ax  \nonumber\\
&\text{subject to} \quad x_{i} \in \{\pm  1 \}, \quad i\in[n] \nonumber\\
&\quad \quad \quad \quad \quad x^{T}\mathbf{1}=0,
\end{align}
which is a non-convex optimization problem. 
Let $Z = xx^{T}$. Reorganizing~\eqref{oublv-ml-equ1}, 
\begin{align}
\label{oublv-relax}
\Zsdp =&  \underset{Z}{\arg\max} ~\langle Z, \Tone  B+\Ttwo A \rangle \nonumber\\
&\text{subject to} \quad Z = xx^{T} \nonumber\\
&\quad \quad \quad \quad \quad Z_{ii}=1, \quad i\in[n] \nonumber\\
&\quad \quad \quad \quad \quad \langle Z, \mathbf{J}\rangle = 0.
\end{align}
By relaxing the rank-one constraint on $Z$, we obtain the following semidefinite programming relaxation of~\eqref{oublv-relax}:
\begin{align}
\label{oublv-sdp-equ2}
\Zsdp =&  \underset{Z}{\arg\max} ~ \langle Z, \Tone  B+\Ttwo A \rangle \nonumber\\
&\text{subject to} \quad Z \succeq 0 \nonumber\\
&\quad \quad \quad \quad \quad Z_{ii}=1, \quad i\in[n] \nonumber\\
&\quad \quad \quad \quad \quad \langle Z, \mathbf{J}\rangle = 0.
\end{align}
For convenience define 
\begin{align*}
    \eta_1 (\qbf,\rho) \triangleq  \frac{\rho}{2} (\sqrt{\qzero }-\sqrt{\qtwo } )^{2} + \frac{1-\rho}{2} (\sqrt{\qone }-\sqrt{\qthree} )^{2} ,
\end{align*}
where $\qbf \triangleq [\qzero, \qtwo, \qone, \qthree]$.

\begin{Theorem}
\label{Theorem-SDP-oublv}
Under the \twomodel{} with binary alphabet where the latent variable $y$ has been revealed, if
\begin{align*}
\begin{cases}
\eta_1 (\qbf,\rhoy) >1 & \text{when} \quad \rhoy \leq 0.5 \\
\eta_1 (\qbf,1-\rhoy) >1 & \text{when} \quad \rhoy > 0.5
\end{cases} 
\end{align*}
then the semidefinite programming estimator is asymptotically optimal, i.e., $\Prob(\Zsdp=Z^{*})\geq 1-o(1)$.
Also, if
\begin{align*}
\begin{cases}
\eta_1 (\qbf,\rhoy) < 1 & \text{when} \quad \rhoy \leq 0.5 \\
\eta_1 (\qbf,1-\rhoy) < 1 & \text{when} \quad \rhoy > 0.5
\end{cases} 
\end{align*}
then for any sequence of estimators $\hat{Z}_{n}$, $\Prob(\hat{Z}_{n}=Z^{*})  \rightarrow 0$.
\end{Theorem}
\begin{proof}
See Appendix~\ref{Proof-Theorem-SDP-oublv}.
\end{proof}

\subsubsection{Recovering $x$ when $y$ is unknown}
\label{sec:sbm-two unknown}
Given an observation of the graph $A$, the aim is to exactly recover $x$ while both latent variables $x$ and $y$ are unknown latent variables.
It is assumed that the estimator does not know anything about the auxiliary latent  variable $y$, which its prior distribution is uniform over $\{y: y \in \{\pm 1\}^{n}\}$.
Notice that $x$ is drawn uniformly from $\{x\in\{\pm 1\}^{n} : x^{T}\mathbf{1}=0\}$.
The log-likelihood of $A$ given $x$ and $y$ is
\begin{align*}
\log \Prob  (A|x, y ) =& \frac{\Tone }{8}  y^{T} (A * xx^{T} )y +\frac{ \Ttwo}{8} x^{T}Ax +\frac{\Tthree }{8} y^{T}Ay+ c, 
\end{align*}
where $\Tone  \triangleq  \log  (\frac{\qzero \qthree}{\qone \qtwo } )$, $\Ttwo \triangleq  \log  (\frac{\qzero \qone }{\qtwo \qthree} )$, and $\Tthree  \triangleq \log  (\frac{\qzero \qtwo }{\qone \qthree} )$, as $n  \rightarrow \infty$ and $c$ is a constant. Then
\begin{align*}
\log  \twocolAlignMarker \Prob  (A|x ) \propto \onecolAlignMarker \log \sum_{\mathcal{Y}} \Prob  (A| x, y ) \\
\propto& \log \sum_{\mathcal{Y}} e^{\frac{\Tone }{\Tthree } y^{T} (A * xx^{T} )y +\frac{\Ttwo}{\Tthree } x^{T}Ax + y^{T}Ay} \\
=& \frac{\Tone +\Ttwo}{\Tthree } x^{T}Ax + \sum_{i} \sum_{j} A_{ij} \twocolbreak +\log \sum_{\mathcal{Y}} e^{\frac{\Tone }{\Tthree } y^{T} (A * xx^{T} )y +y^{T}Ay -\frac{\Tone }{\Tthree }x^{T}Ax -\sum_{i} \sum_{j} A_{ij}} .
\end{align*}
Applying the log-sum-exp approximation, the maximum likelihood estimator is
\begin{align}
\label{tublv-map-equ3}
\hat{x} =&  \underset{x}{\arg\max} ~ x^{T}Ax  \nonumber\\
&\text{subject to} \quad x_{i} \in \{\pm  1 \}, \quad i\in[n] \nonumber\\
&\quad \quad \quad \quad \quad x^{T}\mathbf{1}=0 ,
\end{align}
that is a non-convex optimization problem. 
Let $Z = xx^{T}$. Reorganizing \eqref{tublv-map-equ3} yields
\begin{align}
\label{tublv-relax}
\hat{Z} =&  \underset{Z}{\arg\max} ~\langle Z, A \rangle \nonumber\\
&\text{subject to} \quad Z = xx^{T} \nonumber\\
&\quad \quad \quad \quad \quad Z_{ii}=1, \quad i\in[n] \nonumber\\
&\quad \quad \quad \quad \quad \langle Z, \mathbf{J}\rangle = 0.
\end{align}
Relaxing the rank-one constraint on $Z$, we obtain the following semidefinite programming relaxation of \eqref{tublv-relax}:
\begin{align}
\label{tublv-sdp}
\hat{Z} =&  \underset{Z}{\arg\max} ~\langle Z, A \rangle \nonumber\\
&\text{subject to} \quad Z \succeq 0 \nonumber\\
&\quad \quad \quad \quad \quad Z_{ii}=1, \quad i\in[n] \nonumber\\
&\quad \quad \quad \quad \quad \langle Z, \mathbf{J}\rangle = 0.
\end{align}
For convenience define
\begin{align*}
    \eta_2 &(\qbf,\rho) \triangleq \frac{1}{2}  \bigg( \sqrt{\qzero \rho+\qone (1-\rho)}-\sqrt{\qtwo \rho+\qthree(1-\rho)} \bigg)^2.
\end{align*}

\begin{Theorem}
\label{Theorem-SDP-tublv}
Under the \twomodel{} with binary alphabet, if 
\begin{align*}
\min{\{\eta_2 (\qbf,\rhoy), \eta_2 (\qbf,1-\rhoy)\}} > 1 ,
\end{align*}
then the semidefinite programming estimator is asymptotically optimal, i.e., $\Prob(\Zsdp=Z^{*})\geq 1-o(1)$.
Also, if 
\begin{align*}
\min{\{\eta_2 (\qbf,\rhoy), \eta_2 (\qbf,1-\rhoy)\}} < 1 ,
\end{align*}
then for any sequence of estimators $\hat{Z}_{n}$, $\Prob(\hat{Z}_{n}=Z^{*})  \rightarrow 0$.
\end{Theorem}
\begin{proof}
See Appendix~\ref{Proof-Theorem-SDP-tublv}.
\end{proof}

\begin{Remark}
The results of Theorems~\ref{Theorem-SDP-oublv} and~\ref{Theorem-SDP-tublv} are consistent with Theorems~\ref{theorem-general-side} and~\ref{theorem-general-no side}, respectively.
\end{Remark}

\begin{Remark}
The constraint $x^{T}\mathbf{1}=0$ that has been considered for this part results in a well-defined phase transition threshold for exact recovery of latent variable $x$. In general, $x$ may be a random variable which is drawn uniformly from $\{x \in \{\pm 1\}^{n}: x^{T}\mathbf{1}= (2\rho_{x}-1) n\}$, where $\rho_{x} \triangleq  \frac{1}{n}|\{v \in [n]: x_{v} = 1\}|$. Then $x^{T}\mathbf{1}=0$ is substituted by $x^{T}\mathbf{1}= (2\rho_{x}-1) n$ in semidefinite programming relaxations~\eqref{oublv-sdp-equ2} and~\eqref{tublv-sdp}. Also, due to the robustness of \SDP{}, an approximation of $\rho_{x}$ can be replaced for recovering the latent variable $x$. Investigating the constraint $x^{T}\mathbf{1}= (2\rho_{x}-1)n$ and the robustness of semidefinite programming are beyond the scope of this paper. 
\end{Remark}

\subsection{\Twocensoredmodel{}}
We highlight the specifics of a \twocensoredmodel{} for the purposes of upcoming calculations.
Let $P(k ;\qzero, \xi)$ be a discrete probability density function with parameters $\qzero>0$ and $\xi \in [0, 1]$ as,
\begin{align*}
    P(k ;\qzero, \xi) \triangleq& \xi\qzero \frac{\log n}{n} \delta[k-1] +  (1-\xi)\qzero \frac{\log n}{n} \delta[k+1] \\
    &+ \Big (1-\qzero \frac{\log n}{n} \Big) \delta[k] ,
\end{align*}
where $\delta$ is Dirac delta function. 
The probability of an edge drawn between two nodes $v, \vprim$ is characterized by constants $\qzero, \qtwo, \qone, \qthree$ and $\xi$ such that:
\begin{align*}
A_{ij} \sim  \begin{cases}
P(k ;\qzero, 1-\xi)   & \text{if} \quad x_{v}=x_{\vprim} , y_{v}=y_{\vprim} \\
P(k ;\qtwo, \xi)  & \text{if} \quad x_{v}\neq x_{\vprim},  y_{v}=y_{\vprim} \\
P(k ;\qone, \xi)  & \text{if} \quad x_{v}=x_{\vprim}, y_{v}\neq y_{\vprim}\\
P(k ;\qthree, \xi)  & \text{if} \quad x_{v}\neq x_{\vprim}, y_{v}\neq y_{\vprim}
\end{cases} .
\end{align*}
The corresponding matrix $Q$, as defined earlier, is the same as~\eqref{equ:Q}. 
Also, in this case, the corresponding matrix $\Xi$ will be:
\begin{align}
\label{equ:G}
\Xi =\begin{bmatrix}
(1-\xi) & \xi  & \xi  & \xi \\ 
\xi  & (1-\xi) & \xi & \xi \\ 
\xi  & \xi & (1-\xi) & \xi \\ 
\xi & \xi  & \xi  & (1-\xi)
\end{bmatrix} .
\end{align}

\subsubsection{Recovering $x$ when $y$ is known}
Given an observation of the graph $A$ and $y$ which corresponds to the observed graph, the latent variable $x_{v}$ is recovered exactly for each node $v \in [n]$. In this part, $y$ is considered as an observation which helps the estimator to recover the desired latent variable $x$.
Let 
\begin{align*}
    R \triangleq T A +T (A*W)  +\Tone  (A*A*W) +\Ttwo (A*A) , 
\end{align*}
where $T \triangleq \log \big( \frac{1-\xi}{\xi}\big)$ and $W \triangleq yy^{T}$. 
Since $x$ is chosen uniformly over $\{x\in\{\pm 1\}^{n} : x^{T}\mathbf{1}=0\}$, the maximum likelihood estimator gives the optimal solution. 
Similar to Section~\ref{subsec:sbm-one unknown}, it can be shown that the semidefinite programming relaxation of maximum likelihood estimator for this configuration is
\begin{align}
\label{oublv-sdp-equ2-censored}
\Zsdp =&  \underset{Z}{\arg\max} ~ \langle Z, R \rangle \nonumber\\
&\text{subject to} \quad Z \succeq 0 \nonumber\\
&\quad \quad \quad \quad \quad Z_{ii}=1, \quad i\in[n] \nonumber\\
&\quad \quad \quad \quad \quad \langle Z, \mathbf{J}\rangle = 0 .
\end{align}
For convenience define
\begin{align*}
    \gbf &\triangleq [(1-\xi)\qzero, \xi\qtwo, \xi \qone, \xi \qthree]  , \\
    \hbf &\triangleq [\xi \qzero, (1-\xi) \qtwo, (1-\xi) \qone, (1-\xi) \qthree] .
\end{align*}

\begin{Theorem}
\label{Theorem-SDP-oublv-censored}
Under the \twocensoredmodel{} with binary alphabet where the latent variable $y$ has been revealed, if
\begin{align*}
\begin{cases}
\eta_1(\gbf,\rhoy) +\eta_1(\hbf,\rhoy) >1 & \text{when} \quad \rhoy \leq 0.5 \\
\eta_1(\gbf,1-\rhoy) +\eta_1(\hbf,1-\rhoy) >1 & \text{when} \quad \rhoy > 0.5
\end{cases} 
\end{align*}
then the semidefinite programming estimator is asymptotically optimal, i.e., $\Prob(\Zsdp=Z^{*})\geq 1-o(1)$.
Also, if
\begin{align*}
\begin{cases}
\eta_1(\gbf,\rhoy) +\eta_1(\hbf,\rhoy)  < 1 & \text{when} \quad \rhoy \leq 0.5 \\
\eta_1(\gbf,1-\rhoy) +\eta_1(\hbf,1-\rhoy) < 1 & \text{when} \quad \rhoy > 0.5
\end{cases} 
\end{align*}
then for any sequence of estimators $\hat{Z}_{n}$, $\Prob(\hat{Z}_{n}=Z^{*})  \rightarrow 0$.
\end{Theorem}
\begin{proof}
See Appendix~\ref{Proof-Theorem-SDP-oublv-censored}.
\end{proof}

\subsubsection{Recovering $x$ when $y$ is unknown}
Given an observation of the graph $A$, the aim is to exactly recover $x$ while both latent variables $x$ and $y$ are unknown.
It is assumed that the estimator does not know anything about the auxiliary latent  variable $y$, which its prior distribution is uniform over $\{y: y \in \{\pm 1\}^{n}\}$.
Notice that $x$ is drawn uniformly from $\{x\in\{\pm 1\}^{n} : x^{T}\mathbf{1}=0\}$.
Similar to Section~\ref{sec:sbm-two unknown},
it can be shown that for this configuration the semidefinite programming relaxation of the maximum likelihood estimator is
\begin{align}
\label{tublv-sdp-censored}
\hat{Z} =&  \underset{Z}{\arg\max} ~\langle Z, T A + \Ttwo (A*A) \rangle \nonumber\\
&\text{subject to} \quad Z \succeq 0 \nonumber\\
&\quad \quad \quad \quad \quad Z_{ii}=1, \quad i\in[n] \nonumber\\
&\quad \quad \quad \quad \quad \langle Z, \mathbf{J}\rangle = 0.
\end{align}

\begin{Theorem}
\label{Theorem-SDP-tublv-censored}
Under the \twocensoredmodel{} with binary alphabet, if 
\begin{align*}
\min \big \{ \eta_2(\gbf,\rhoy)+\eta_2(\hbf,\rhoy), \eta_2(\gbf,1-\rhoy)+\eta_2(\hbf,1-\rhoy) \big \} > 1 ,
\end{align*}
then the semidefinite programming estimator is asymptotically optimal, i.e., $\Prob(\Zsdp=Z^{*})\geq 1-o(1)$.
Also, if 
\begin{align*}
\min \big \{\eta_2(\gbf,\rhoy)+\eta_2(\hbf,\rhoy), \eta_2(\gbf,1-\rhoy)+\eta_2(\hbf,1-\rhoy) \big \} < 1 ,
\end{align*}
then for any sequence of estimators $\hat{Z}_{n}$, $\Prob(\hat{Z}_{n}=Z^{*})  \rightarrow 0$.
\end{Theorem}
\begin{proof}
See Appendix~\ref{Proof-Theorem-SDP-tublv-censored}.
\end{proof}
\begin{Remark}
The results of Theorems~\ref{Theorem-SDP-oublv-censored} and~\ref{Theorem-SDP-tublv-censored} are consistent with Theorems~\ref{theorem-general-cencored-side} and~\ref{theorem-general-cencored-no side}, respectively.
\end{Remark}

% \begin{Remark}
% The semidefinite programming optimizations of this section can be simulated by the cvx toolbox for graph models with finite number of nodes. It can be shown that the asymptotic results of this section can shed light on the community detection over graphs of finite size.   
% \end{Remark}
%%%%%%%%%%%%%%%%%%%%%%%%%%%%%%%%%%%%%%
\section{Discussion \& Numerical Results}

It is illuminating to review the flow of the development of the achievability results througout this paper:
%In this paper, the following steps are taken to prove the achievability for Theorems~\ref{Theorem-SDP-oublv}, ~\ref{Theorem-SDP-tublv}, ~\ref{Theorem-SDP-oublv-censored}, and~\ref{Theorem-SDP-tublv-censored}: 
\begin{enumerate}
    \item Calculate the Lagrangian of the corresponding optimization
    \item Extract the dual optimal solution based on the Lagrange multipliers
    \item Show that $\Zsdp=Z^{*}$ is primal optimal solution
	\item Show that $\Zsdp=Z^{*}$ is unique
	\item Extract the conditions under which the dual optimal solution holds
\end{enumerate}
The converses follow the following sequence:
%To prove the converse of Theorems~\ref{Theorem-SDP-oublv}, ~\ref{Theorem-SDP-tublv}, ~\ref{Theorem-SDP-oublv-censored}, and~\ref{Theorem-SDP-tublv-censored}, we stick to the following steps:
\begin{enumerate}
    \item Extract the maximum likelihood estimator
    \item Extract the conditions under which the maximum likelihood estimator fails 
\end{enumerate}

\begin{figure}
    \centering
    \includegraphics[width=\Figwidth]{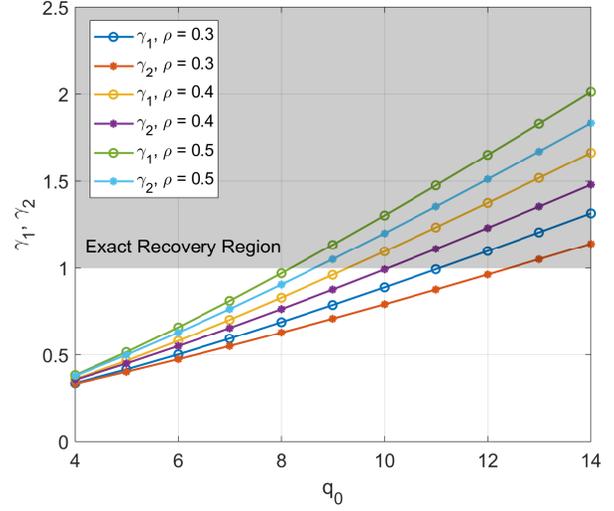}
    \caption{Exact recovery region of $x$ in the context of Eq.~\eqref{equ:Q}, with $\qone=3, \qtwo=\qthree=1$.}
    \label{Fig_threshold_6_12}
\end{figure}

\begin{figure}
    \centering
    \includegraphics[width=\Figwidth]{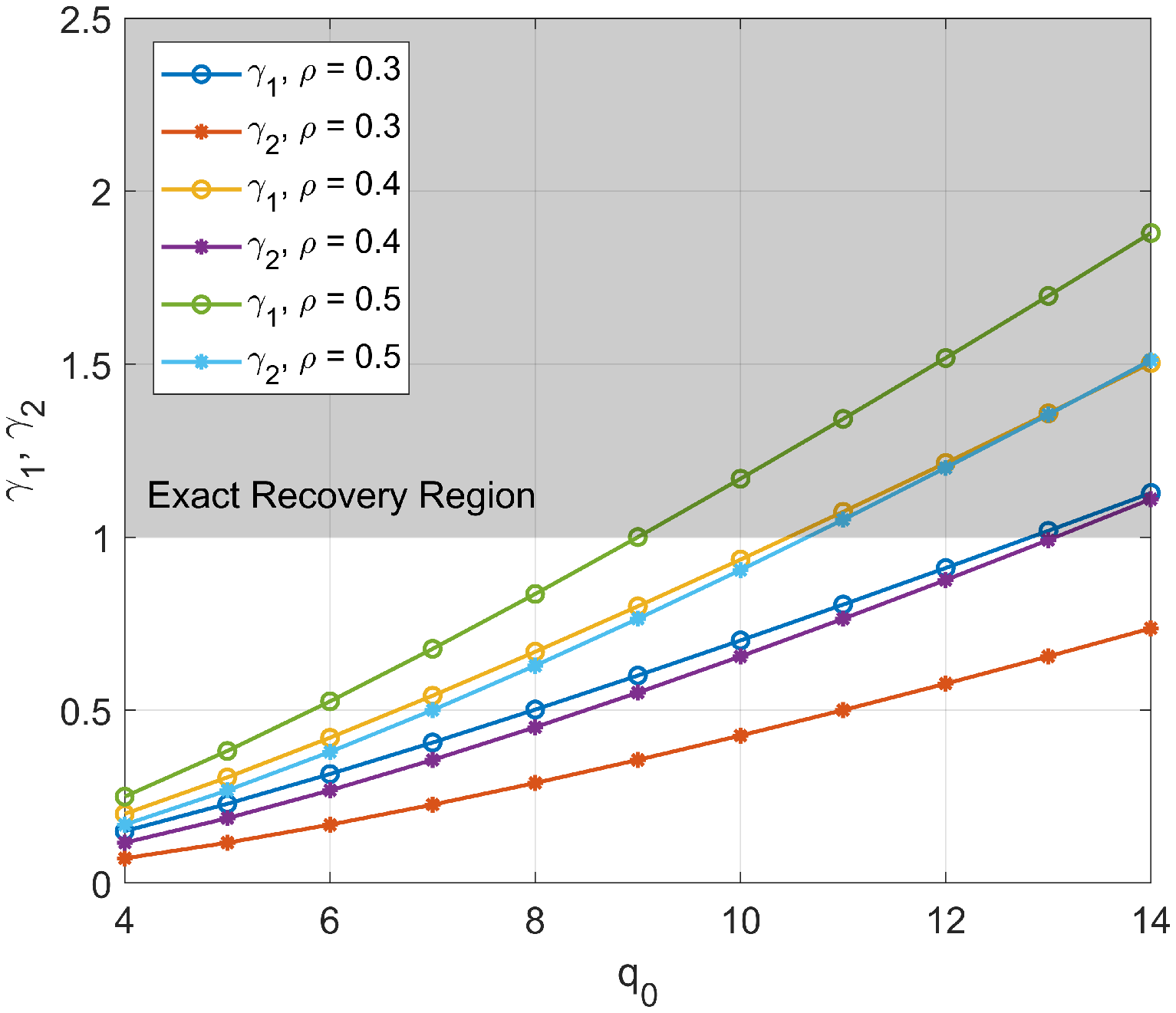}
    \caption{Exact recovery region of $x$ in the context of Eq.~\eqref{equ:Q}, with $\qtwo=\qone=\qthree=1$.}
    \label{Fig_threshold_8_34}
\end{figure}
%%%%%%%%%%%%%%%%%%%%%%%%%%%%%%%%%%%%%
\begin{figure}[ht]
    \centering
    \includegraphics[width=\Figwidth]{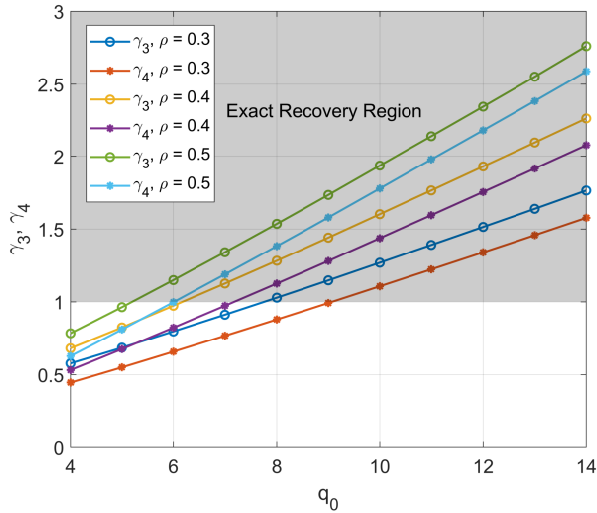}
    \caption{Exact recovery region of $x$ in the context of Eq.~\eqref{equ:Q} and Eq.~\eqref{equ:G}, with $\xi=0.1$, $\qone=3$, and $\qtwo=\qthree=1$.}
    \label{Fig_threshold_6_12-censored}
\end{figure}
\begin{figure}[ht]
    \centering
    \includegraphics[width=\Figwidth]{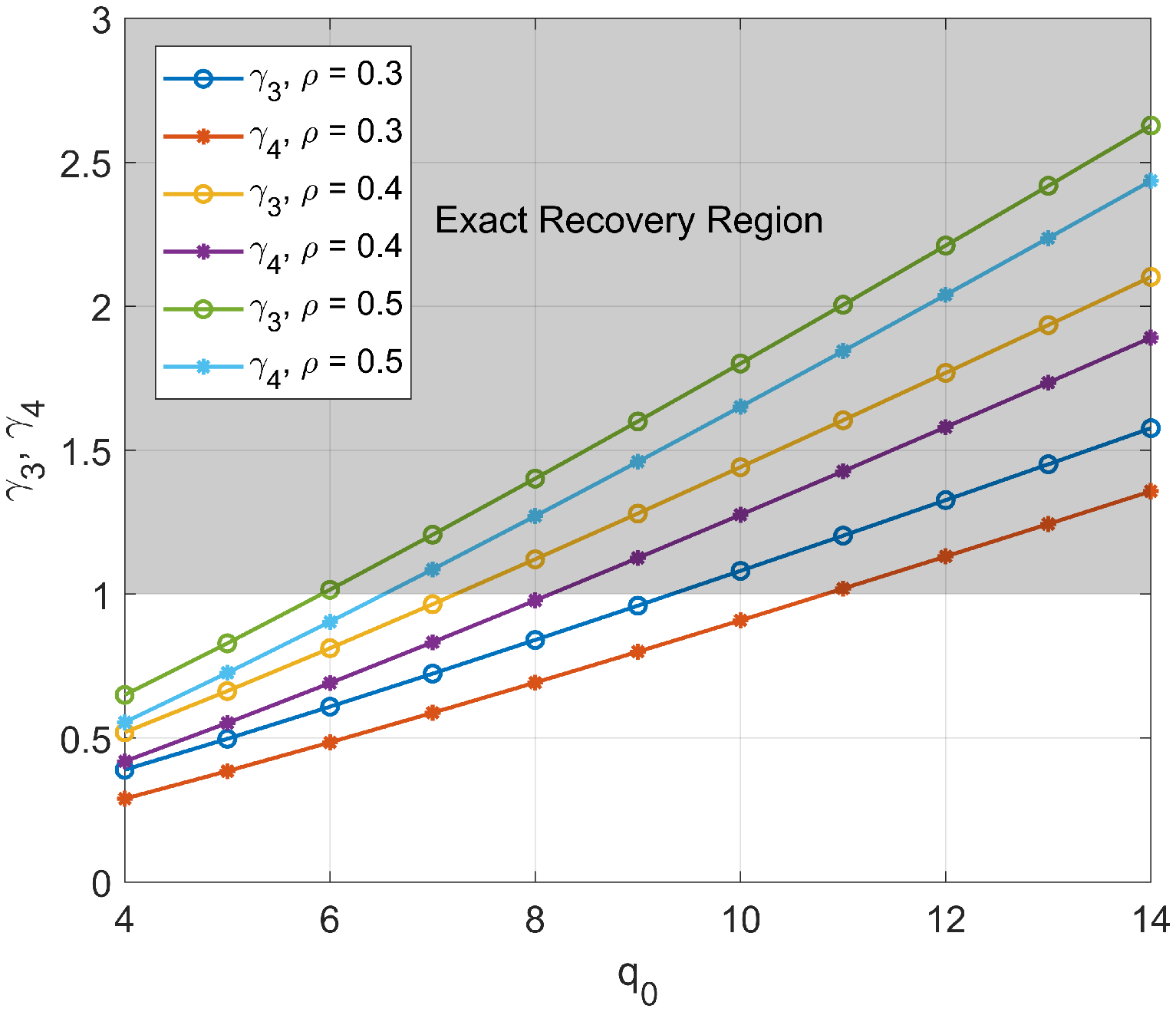}
    \caption{Exact recovery region of $x$ in the context of Eq.~\eqref{equ:Q} and Eq.~\eqref{equ:G}, with $\xi = 0.1$, and $\qtwo=\qone=\qthree=1$.}
    \label{Fig_threshold_8_34-censored}
\end{figure}
%%%%%%%%%%%%%%%%%%%%%%%%%%%%%%%%%%%%%
To give a pictorial view of some results of the paper, we plot some results in the context of the \twomodel{} represented by~\eqref{equ:Q} and \twocensoredmodel{} represented by~\eqref{equ:Q} and~\eqref{equ:G}. For ease of notation, we define 
\begin{align*}
    &\gamma_1 \triangleq \min \{\eta_1(\qbf,\rho),\eta_1(\qbf,1-\rho)\}, \\
    & \gamma_2 \triangleq \min \{\eta_2(\qbf,\rho),\eta_2(\qbf,1-\rho)\} , \\
    &\gamma_3 \triangleq \min \{\eta_1(\gbf,\rho)+\eta_1(\hbf,\rho),\eta_1(\gbf,1-\rho) +\eta_1(\hbf,1-\rho) \} , \\
    &\gamma_4 \triangleq \min \{\eta_2(\gbf,\rho) +\eta_2(\hbf,\rho),\eta_2(\gbf,1-\rho) +\eta_2(\hbf,1-\rho) \} .
\end{align*}

For the \twomodel{}, Figures~\ref{Fig_threshold_6_12} and~~\ref{Fig_threshold_8_34} 
show the exact recovery region for recovering the latent variable $x$ when the secondary latent variable $y$ is either known or unknown.
The curves in these figures are based on the obtained results in Theorem~\ref{Theorem-SDP-oublv} and Theorem~\ref{Theorem-SDP-tublv}. 
These figures encompass several curves plotted for different values of $\qzero$, $\qtwo$, $\qone$, $\qthree$ in~\eqref{equ:Q}, and $\rhoy$. At each figure, we consider fixed values for $\qtwo$, $\qone$, $\qthree$ and  vary the values of $\qzero$ and $\rhoy$. 
A comparison between the curves in Figures~\ref{Fig_threshold_6_12} and \ref{Fig_threshold_8_34} clarifies the role of the revealed latent variable $y$ for recovering the desired latent variable $x$.

For the \twocensoredmodel{}, Figures~\ref{Fig_threshold_6_12-censored} and~~\ref{Fig_threshold_8_34-censored} 
show the exact recovery region for recovering the latent variable $x$ when the secondary latent variable $y$ is either known or unknown.
The curves in these figures are based on the obtained results in Theorem~\ref{Theorem-SDP-oublv-censored} and Theorem~\ref{Theorem-SDP-tublv-censored}. 
These figures consist of several curves plotted for different values of $\qzero$, $\qtwo$, $\qone$, $\qthree$ in~\eqref{equ:Q} and $\rhoy$, while $\xi=0.1$ in~\eqref{equ:G}. At each figure, we consider fixed values for $\xi$, $\qtwo$, $\qone$, $\qthree$ and  vary the values of $\qzero$ and $\rhoy$. 
A comparison between the curves in Figures~\ref{Fig_threshold_6_12-censored} and \ref{Fig_threshold_8_34-censored} clarifies the role of the revealed latent variable $y$ for recovering the desired latent variable $x$.

To gain an understanding of the scope of our asymptotic results, under the conditions of Figures~\ref{Fig_threshold_6_12} and~\ref{Fig_threshold_6_12-censored}, we performed several simulations on $10^4$ graph realizations with various graph sizes obtained from the proposed models in Section~\ref{system-model}. The obtained average error probability (AEP) is around $10^{-5}$ in the regimes just inside the region of exact recovery, and around $10^{-2}$ in the regimes just outside the region of exact recovery. The details of these simulations are represented in Tables~\ref{SBM} and~\ref{SBM_1}. At each simulation, we consider fixed values for $\qtwo$, $\qone$, $\qthree$ and  vary the values of $\qzero$, $\rhoy$, and $n$. 
% % \footnote{Under the conditions of Figure~\ref{Fig_threshold_6_12} (\ref{Fig_threshold_6_12-censored}), other simulation parameter $q_0$ is chosen such that $\gamma_1 =\gamma_2= 1.25$ ($\gamma_3 =\gamma_4= 1.25$) inside the exact recovery regime and $\gamma_1 =\gamma_2= 0.75$ ($\gamma_3 =\gamma_4= 0.75$) outside the exact recovery regime.}

\begin{table}[]
\centering
\begin{tabular}{@{}cc|cc|ccc@{}}
\toprule
\multirow{2}{*}{$y$} & \multirow{2}{*}{$n$} & \multicolumn{2}{c|}{BSBM} & \multicolumn{3}{c}{BCBM} \\ \cmidrule(l){3-7} 
                   &                    & $q_0$          & AEP         & $q_0$     & $\xi$     & AEP    \\ \midrule
\text{Known}                   & 100                    &  7           & $3.8 \times 10^{-2}$            & 4       &  0.1      &    $6.7 \times 10^{-2}$    \\
\text{Known}                   & 200                    &  7           & $2.4 \times 10^{-2}$            & 4       & 0.1       &    $4.9 \times 10^{-2}$    \\ 
\text{Known}                   & 300                    &  7           & $1.9 \times 10^{-2}$            & 4       & 0.1       &    $3.6 \times 10^{-2}$    \\
\text{Known}                   & 400                    &  7           & $1.5 \times 10^{-2}$            & 4       & 0.1       &    $2.5 \times 10^{-2}$    \\
\text{Known}                   & 500                    &  7           & $1.1 \times 10^{-2}$            & 4       & 0.1       &    $1.6 \times 10^{-2}$    \\
\text{Known}                   & 100                    & 9            & $8.1 \times 10^{-5}$            & 6       & 0.1       &    $4.6 \times 10^{-5}$    \\
\text{Known}                   & 200                    & 9            &  $5.9 \times 10^{-5}$           & 6       & 0.1       &    $3.2 \times 10^{-5}$    \\ 
\text{Known}                   & 300                    & 9            &  $4.2 \times 10^{-5}$           & 6       & 0.1       &    $2.4 \times 10^{-5}$    \\
\text{Known}                   & 400                    & 9            & $2.8 \times 10^{-5}$            & 6       & 0.1       &    $1.7 \times 10^{-5}$    \\
\text{Known}                   & 500                    & 9            & $1.8 \times 10^{-5}$            & 6       & 0.1       &    $1.2 \times 10^{-5}$    \\
\text{Unknown}                   & 100                    & 8            &  $5.7 \times 10^{-2}$           &  5      & 0.1       & $5.6 \times 10^{-2}$       \\
\text{Unknown}                   & 200                    & 8            & $4.1 \times 10^{-2}$            & 5       &  0.1      & $3.9 \times 10^{-2}$       \\ 
\text{Unknown}                   & 300                    & 8            & $2.7 \times 10^{-2}$            & 5       & 0.1       & $2.5 \times 10^{-2}$       \\
\text{Unknown}                   & 400                    & 8            & $1.8 \times 10^{-2}$            & 5       & 0.1       & $1.6 \times 10^{-2}$        \\
\text{Unknown}                   & 500                    &  8           & $1.3 \times 10^{-2}$            & 5       & 0.1       & $1.1 \times 10^{-2}$        \\
\text{Unknown}                   & 100                    & 10            &  $6.2 \times 10^{-5}$           & 7       & 0.1       & $6.3 \times 10^{-5}$       \\
\text{Unknown}                   & 200                    & 10            & $4.4 \times 10^{-5}$             & 7       & 0.1       & $4.0 \times 10^{-5}$       \\ 
\text{Unknown}                   & 300                    & 10            & $3.3 \times 10^{-5}$            & 7       & 0.1       & $2.3 \times 10^{-5}$       \\
\text{Unknown}                   & 400                    & 10            &  $2.3 \times 10^{-5}$           &  7      & 0.1       & $1.7 \times 10^{-5}$       \\
\text{Unknown}                   & 500                    & 10            &  $1.4 \times 10^{-5}$           & 7       & 0.1       & $1.3 \times 10^{-5}$       \\ \bottomrule
\end{tabular}
\caption{Semidefinite programming optimization of (8) and (10), with $\qone=3, \qtwo=\qthree=1$, and $\rho=0.5$.}
\label{SBM}
\end{table}

\begin{table}[]
\centering
\begin{tabular}{@{}cc|cc|ccc@{}}
\toprule
\multirow{2}{*}{$y$} & \multirow{2}{*}{$n$} & \multicolumn{2}{c|}{BSBM} & \multicolumn{3}{c}{BCBM} \\ \cmidrule(l){3-7} 
                   &                    & $q_0$          & AEP         & $q_0$     & $\xi$     & AEP    \\ \midrule
\text{Known}                   & 100                    &  10           & $7.6 \times 10^{-2}$            & 7       &  0.1      &    $4.1 \times 10^{-2}$    \\
\text{Known}                   & 200                    &  10           & $5.1 \times 10^{-2}$            & 7       & 0.1       &    $3.1 \times 10^{-2}$    \\ 
\text{Known}                   & 300                    &  10           & $3.0 \times 10^{-2}$            & 7       & 0.1       &    $2.3 \times 10^{-2}$    \\
\text{Known}                   & 400                    &  10           & $2.1 \times 10^{-2}$            & 7       & 0.1       &    $1.8 \times 10^{-2}$    \\
\text{Known}                   & 500                    &  10           & $1.3 \times 10^{-2}$            & 7       & 0.1       &    $1.3 \times 10^{-2}$    \\
\text{Known}                   & 100                    & 12            & $6.7 \times 10^{-5}$            & 9       & 0.1       &    $3.9 \times 10^{-5}$    \\
\text{Known}                   & 200                    & 12            &  $5.1 \times 10^{-5}$           & 9       & 0.1       &    $2.5 \times 10^{-5}$    \\ 
\text{Known}                   & 300                    & 12            &  $3.6 \times 10^{-5}$           & 9       & 0.1       &    $1.8 \times 10^{-5}$    \\
\text{Known}                   & 400                    & 12            & $2.5 \times 10^{-5}$            & 9       & 0.1       &    $1.2 \times 10^{-5}$    \\
\text{Known}                   & 500                    & 12            & $1.6 \times 10^{-5}$            & 9       & 0.1       &    $1.0 \times 10^{-5}$    \\
\text{Unknown}                   & 100                    & 11            &  $4.3 \times 10^{-2}$           &  8      & 0.1       & $4.2 \times 10^{-2}$       \\
\text{Unknown}                   & 200                    & 11            & $3.3 \times 10^{-2}$            & 8       &  0.1      & $2.9 \times 10^{-2}$       \\ 
\text{Unknown}                   & 300                    & 11            & $2.4 \times 10^{-2}$            & 8       & 0.1       & $2.0 \times 10^{-2}$       \\
\text{Unknown}                   & 400                    & 11            & $1.7 \times 10^{-2}$            & 8       & 0.1       & $1.3 \times 10^{-2}$        \\
\text{Unknown}                   & 500                    &  11           & $1.2 \times 10^{-2}$            & 8       & 0.1       & $1.0 \times 10^{-2}$        \\
\text{Unknown}                   & 100                    & 13            &  $4.2 \times 10^{-5}$           & 10       & 0.1       & $4.8 \times 10^{-5}$       \\
\text{Unknown}                   & 200                    & 13            & $2.6 \times 10^{-5}$             & 10       & 0.1       & $3.3 \times 10^{-5}$       \\ 
\text{Unknown}                   & 300                    & 13            & $1.7 \times 10^{-5}$            & 10       & 0.1       & $2.2 \times 10^{-5}$       \\
\text{Unknown}                   & 400                    & 13            &  $1.3 \times 10^{-5}$           &  10      & 0.1       & $1.5 \times 10^{-5}$       \\
\text{Unknown}                   & 500                    & 13            &  $1.1 \times 10^{-5}$           & 10       & 0.1       & $1.0 \times 10^{-5}$       \\ \bottomrule
\end{tabular}
\caption{Semidefinite programming optimization of (8) and (10), with $\qone=3, \qtwo=\qthree=1$, and $\rho=0.3$.}
\label{SBM_1}
\end{table}

\section{Conclusion}
This paper presents and analyzes a new generalization of the stochastic and censored block models in which, in addition to the latent variable representing community labels, there exists another (secondary) latent variables that are not part of community detection. These secondary latent variables may be known, unknown, or partially known. This model represents community detection problems where the community labels alone does not explain all the dependencies between the graph edges.  
 
We investigate the exact recovery threshold for these models under maximum likelihood detection, and also  analyze a semidefinite programming algorithm for recovering the desired latent variable under the \twomodel{} and the \twocensoredmodel{} for both scenarios. 

%%%%%%%%%%%%%%%%%%%%%%%%%%%%%%%%%%
\appendices
%%%%%%%%%%%%%%%%%%%%%%%%%%%%%%%%%%
\section{Proof of Lemma~\ref{Lemma-Poisson}}
\label{Proof-Lemma-Poisson}
Define
\begin{align*}
    f_{1}(t) &\triangleq \prod_{i=1}^{m} \bigg(\frac{b_{i}}{a_{i}}\bigg)^{(t-1)d_{i}} e^{(t-1)(a_{i}-b_{i})} ,\\
    f_{2}(t) &\triangleq \prod_{i=1}^{m} \bigg(\frac{b_{i}}{a_{i}}\bigg)^{td_{i}} e^{t(a_{i}-b_{i})}.
\end{align*}

For any $t \in [0,1]$,
\begin{align*}
\sum_{d \in \mathbb{Z}_{+}^{m}} & \min \{ \mathcal{P}_{a}(d) \paaa , \mathcal{P}_{b}(d) \pbbb \} \\
\leq& \max \{\paaa, \pbbb \} \sum_{d\in {\mathbb{Z}_{+}^{m}}} \min \{ \mathcal{P}_{a}(d) , \mathcal{P}_{b}(d) \} \\
=& \max \{\paaa, \pbbb \}  \exp \bigg(-\sum_{i} \big[ta_{i}+(1-t)b_{i}-a_{i}^{t}b_{i}^{1-t} \big] \bigg) \\ &\times \sum_{d \in \mathbb{Z}_{+}^{m}} \prod_{i} \frac{ (a_{i}^{t}b_{i}^{1-t})^{d_{i}} }{d_{i}!} e^{-a_{i}^{t}b_{i}^{1-t}}  \min \{ f_{1}(t), f_{2}(t)  \}.
\end{align*}
Both $f_{1}(t)$ and $f_{2}(t)$ are monotonic and $\frac{f_2(t)}{f_1(t)}$ is a positive constant (does not depend on $t$), thus $\min\{f_1,f_2\}$ is also monotonic in $t$. Since $f_1(1)=f_2(0)=1$, for all $t$ we have:
\begin{align*}
\min \{ f_{1}(t), f_{2}(t) \} \leq 1. 
\end{align*}
Notice that
\begin{align*}
\sum_{d \in \mathbb{Z}_{+}^{m}} \prod_{i} \frac{ (a_{i}^{t}b_{i}^{1-t})^{d_{i}} }{d_{i}!} e^{-a_{i}^{t}b_{i}^{1-t}}  = 1. 
\end{align*}
Then
\begin{align}
\label{two-Poisson-upper}
I(a,b) \leq  \max \{\paaa, \pbbb \}  e^{-\sum_{i=1}^{m} \big[ta_{i}+(1-t)b_{i}-a_{i}^{t}b_{i}^{1-t} \big]}.
\end{align}
For the value of $t$ that maximizes the right-hand side of inequality~\eqref{two-Poisson-upper}, we have
\begin{align*}
\sum_{d \in \mathbb{Z}_{+}^{m}} \min \{ \mathcal{P}_{a}(d) \paaa , \mathcal{P}_{b}(d) \pbbb \} &\leq \max \{\paaa, \pbbb \} e^{-\divergence(a,b)}.
\end{align*}
Notice that $t^{*}$ satisfies 
\begin{align*}
\prod_{i=1}^{m} \bigg(\frac{b_{i}}{a_{i}}\bigg)^{a_{i}^{t^{*}}b_{i}^{1-t^{*}}} e^{a_{i}-b_{i}} = 1 .
\end{align*}
Then at the optimal $t^{*}$,
\begin{align*}
\sum_{d \in \mathbb{Z}_{+}^{m}} & \min \{ \mathcal{P}_{a}(d) \paaa , \mathcal{P}_{b}(d) \pbbb \} \nonumber\\
\geq& \min \{\paaa, \pbbb \} \sum_{d \in \mathbb{Z}_{+}^{m}} \min \{ \mathcal{P}_{a}(d) , \mathcal{P}_{b}(d) \} \nonumber\\
\overset{(a)}{\geq} & \min \{\paaa, \pbbb \}  e^{-\divergence(a,b)}  \prod_{i} \frac{ (a_{i}^{t^{*}}b_{i}^{1-t^{*}})^{a_{i}^{t^{*}}b_{i}^{1-t^{*}}} }{a_{i}^{t^{*}}b_{i}^{1-t^{*}}!} e^{-a_{i}^{t^{*}}b_{i}^{1-t^{*}}} \\
\overset{(b)}{\geq} & \min \{\paaa, \pbbb \}  e^{-\divergence(a,b)} \prod_{i} \frac{1}{e} \big(a_{i}^{t^{*}}b_{i}^{1-t^{*}}\big)^{-\frac{1}{2}} ,
\end{align*}
where $(a)$ holds because 
\begin{align*}
\sum_{d \in \mathbb{Z}_{+}^{m}} \min \{ \mathcal{P}_{a}(d) , \mathcal{P}_{b}(d) \} \geq \min \{ \mathcal{P}_{a}(d^*) , \mathcal{P}_{b}(d^*) \},
\end{align*}
where $d^*$ is defined by $d_i^* \triangleq a_i^{t^*}b_i^{1-t^*}$, and $(b)$ is due to Stirling's approximation $n! \leq n^{n+\frac{1}{2}} e^{-n +1}$ for any $n\geq 1$.
%%%%%%%%%%%%%%%%%%%%%%%%%%%%%%%%%%
\section{Proof of Theorem~\ref{theorem-general-side}}
\label{proof-theorem-general-side}
We aim to recover $x_{v}$ when $y_{v}$ is known.
Given a realization of $D$ and $y_{v}$, our goal is to minimize the error probability by selecting the most likely hypothesis, i.e.,
\begin{align*}
\argmax_{i} \Prob\{H_i| D = d,y_v\}, 
\end{align*}
or equivalently, since $d, y_v$ are known observations,
\begin{align*}
\argmax_{i} P(d| H_i, y_{v}) \Prob\{H_i, y_{v}\} ,
\end{align*}
which is the maximum a posteriori (MAP) detector, which we rewrite:
\begin{align}
\label{MapDetector-SideInformation}
\argmax_{i} P(d| H_i,y_{v}) P_{i,y_{v}} .
\end{align}
Solving~\eqref{MapDetector-SideInformation} requires $m_{x}-1$ pairwise comparisons of the hypotheses. From this viewpoint, if 
\begin{align}
P(d| H_i,y_{v}) P_{i,y_{v}} \leq P(d| H_k,y_{v}) P_{k,y_{v}}, 
\label{eq:twohypothesiscorrect}
\end{align}
then a pairwise comparison will choose $H_k$ over $H_i$. Now assume the correct hypothesis is $H_i$, and denote by $\bad_{ik}$ the region of $D$ for which \eqref{eq:twohypothesiscorrect} is satisfied, i.e., $H_i$ has a {\em worse} metric compared with $H_k$. Also denote by $\bad_i$ the region for $D$ where the overall MAP decoder is in error. The dependence of error regions $\bad_{ik}$ and $\bad_i$ on $y_v$ is implicit. Then the probability of error
\begin{align}
\label{error_probability}
P_{e} = \sum_{i} \Prob\{D \in \bad_{i} | H_i,y_v\} P_{i,y_{v}} .
\end{align}
Since $\bad_i \subset \cup_k \bad_{ik}$,
\begin{align*}
P_{e} \le \sum_{i} \sum_{k\neq i} \Prob \{D \in \bad_{ik} | H_i,y_v\} P_{i,y_{v}} .
\end{align*}
From the earlier Poisson assumption
$P(d | H_i, y_v) = \mathcal{P}_{\lambda_{i,y_v}}(d)$ it follows that:
\begin{align*}
\min \{\mathcal{P}_{\lambdavec{i}{y_v}}(d) P_{i,y_{v}} &, \mathcal{P}_{\lambdavec{k}{y_v}}(d) P_{k,y_{v}}  \} = \twocolbreak
\begin{cases}
\mathcal{P}_{\lambdavec{i}{y_v}}(d) \, P_{i,y_{v}}  & \text{when} ~D \in \bad_{ik} \\
\mathcal{P}_{\lambdavec{k}{y_v}}(d) \, P_{k,y_{v}}  & \text{when} ~D \in  \bad_{ik}^c
\end{cases} .
\end{align*}
Therefore, substituting into the union bound:
\begin{align}
\label{error probability-upper bound}
P_{e} & \leq \sum_d \sum_i\sum_{k>i} \min \{ \mathcal{P}_{\lambdavec{i}{y_v}}(d) P_{i,y_{v}}, \mathcal{P}_{\lambdavec{k}{y_v}}(d) P_{k,y_{v}} \}.
\end{align}
For bounding the error probability~\eqref{error probability-upper bound}, it suffices to find an upper bound for
\begin{align}
\label{main-side information}
\sum_{d} \min \{ \mathcal{P}_{\lambdavec{i}{y_v}}(d) P_{i,y_{v}}, \mathcal{P}_{\lambdavec{k}{y_v}}(d) P_{k,y_{v}} \}.
\end{align}
It follows from Lemma~\ref{Lemma-Poisson} that
\begin{align}
\label{final-upper-side}
P_{e} &\leq  \sum_{i} \sum_{k > i} \max \{P_{i,y_{v}}, P_{k,y_{v}} \} e^{-\divergence( \lambdavec{i}{y_v},\lambdavec{k}{y_v})}  \twocolbreak =
 \sum_{i} \sum_{k > i} n^{-\divergence \big( \qij{i}{y_v},  \qij{k}{y_v} \big) + o(1)} . 
\end{align}

We now bound the error probability of decoding rule~\eqref{MapDetector-SideInformation} from below. Since 
\begin{align}
\label{noneqiality-lower bound}
\sum_{k\neq i} &\Prob\{D \in \bad_{ik} | H_i,y_k\} \leq  (m_{x}-1) \Prob\{D \in \bad_i | H_i,y_v\}, 
\end{align}
substituting~\eqref{noneqiality-lower bound} into \eqref{error_probability} yields 
\begin{align*}
P_{e} \geq& \frac{1}{m_{x}-1} \sum_i\sum_{k\neq i} \Prob\{D \in \bad_{ik} | H_i,y_v\} P_{i,y_v} = \frac{1}{m_{x}-1} \\
&\times \sum_{i}\sum_{k > i}\sum_d \min \Big\{ \mathcal{P}_{\lambda_{i,y_v}}(d) P_{i,y_v},  \mathcal{P}_{\lambda_{k,y_v}}(d) P_{k,y_v} \Big\}.
\end{align*}
Then it suffices to find a lower bound for~\eqref{main-side information} to bound the error probability from below. It follows from Lemma~\ref{Lemma-Poisson} that
\begin{align}
\label{final-lower-side}
P_{e} \geq&  \sum_{i} \sum_{k > i} c' \min \{P_{i,y_{v}}, P_{k,y_{v}} \} (\log n)^{-\frac{m}{2}}  e^{-\divergence( \lambdavec{i}{y_v},\lambdavec{k}{y_v})}  \nonumber \\
=&  \sum_{i} \sum_{k > i} n^{-\divergence( \qij{i}{y_v}, \qij{k}{y_v} ) + o(1)}, 
\end{align}
where $c'$ is a constant and $m$ is the number of elements in vector $d$, i.e., the product of alphabet sizes of $x_v$ and $y_v$. 
The lower and upper bounds~\eqref{final-upper-side} and~\eqref{final-lower-side} imply that the true hypothesis is recovered correctly if $\divergence( \qij{i}{y_v}, \qij{k}{y_v} ) > 1$, for a given $y_{v}$ and any $i\neq k$. This means that a known latent variable restricts the number of pairwise comparisons. Then under the \twomodel{} in which the latent variable $y$ is known, and the latent variable $x$ is unknown, exact recovery is possible for $x$ if and only if
\begin{align}
\label{equ: final 1}
\min_j \min_{i\neq k} \divergence(\qij{i}{j},\qij{k}{j}) > 1.
\end{align}
%%%%%%%%%%%%%%%%%%%%%%%%%%%%%%%%%%
\section{Proof of Theorem~\ref{theorem-general-no side}}
\label{proof-theorem-general-no side}

We aim to recover $x_{v}$ when $y_{v}$ is unknown, given a realization of $D$ for node $v$.
For this setting the MAP detector is
\begin{align*}
\argmax_{i} \Prob\{H_i| D = d \}, 
\end{align*} 
or equivalently, 
\begin{align}
\label{map-detector}
\argmax_{i} \sum_{y_v} \prod_{l} P \Bigg( \sum_{j} \dij{l}{j} | H_i, y_v \Bigg) P_{i,y_v}. 
\end{align}
Solving~\eqref{map-detector} requires $m_{x}-1$ pairwise comparisons. In these comparisons, if
\begin{align}
\sum_{y_v} \prod_{l} P &\Bigg( \sum_{j} \dij{l}{j} | H_i, y_v \Bigg) P_{i,y_v} \nonumber \\
&< \sum_{y_v} \prod_{l} P \Bigg( \sum_{j} \dij{l}{j} | H_k, y_v \Bigg) P_{k,y_v} , 
\label{eq:pairwisecomparison}
\end{align}
then we conclude hypothesis $H_i$ is ruled out, i.e., $x_{v} \neq i$, because another hypothesis $H_k$ has a better metric. 
Denote by $\bad_{ik}$ the region of $D$ for which $H_i$ has a {\em worse} metric compared with $H_k$, i.e., the region for $D$ in which \eqref{eq:pairwisecomparison} is satisfied. Also denote by $\bad_i$ the region for $D$ where the overall MAP decoder is in error. 
The error probability of MAP decoder~\eqref{map-detector} is given by
\begin{align}
\label{error-probability-no side}
P_{e} = \sum_{i} \sum_{y_v} \Prob(D \in \bad_i| H_i,y_v\} P_{i,y_v} .
\end{align}
Since $\bad_i \subset \cup_k \bad_{ik}$, via the union bound,
\begin{align}
\label{noneqiality-union bound-no side}
\sum_{y_v} & \Prob\{D \in \bad_i | H_i,y_v\} P_{i,y_v}  \nonumber \\
&\leq \sum_{y_v}\sum_{k\neq i} \Prob\{D \in \bad_{ik} | H_i,y_v\} P_{i,y_v}.
\end{align}
Using the Poisson approximation and the  additive property of Poisson distribution:
\begin{align*}
I(d,i, y_v) &\triangleq  \prod_{l} P \bigg(\sum_{j} \dij{l}{j} | H_i,y_v \bigg) \\
&=  \prod_{l} \mathcal{P}_{\sum_{j} \lambdascalar{i}{y_v}{l}{j}} \bigg(\sum_{j} \dij{l}{j} \bigg) .
\end{align*}
Therefore,
\begin{align*}
\min \Big\{ I(d,i, y_v) P_{i,y_v}, & I(d,k,y_v) P_{k,y_v} \Big\} = \twocolbreak 
\begin{cases}
 I(d,i,y_v) P_{i,y_v} & \text{when} ~D \in \bad_{ik} \\
I(d,k,y_v) P_{k,y_v} & \text{when} ~D \in \bad_{ik}^c
\end{cases} .
\end{align*}
Substituting~\eqref{noneqiality-union bound-no side} into~\eqref{error-probability-no side} yields
\begin{align}
\label{error probability-upper bound- no side}
P_{e} \leq& \sum_{i}\sum_{k\neq i}\sum_{y_v} \Prob\{D \in \bad_{ik} | H_i,y_v\} P_{i,y_v} \nonumber \\
 =& \sum_{i}\sum_{k> i} \sum_{y_v} \sum_{d} \min \Big\{ I(d,i, y_v) P_{i,y_v}, I(d,k,y_v) P_{k,y_v} \Big\}.
\end{align}
For bounding the error probability~\eqref{error probability-upper bound- no side} from above, it suffices to find an upper bound for
\begin{align}
\label{main-no side}
\sum_{d \in \mathbb{Z}_{+}^{m}} \min \Big\{ I(d,i, y_v) P_{i,y_v}, I(d,k,y_v) P_{k,y_v} \Big\}.
\end{align}
Applying Lemma~\ref{Lemma-Poisson} yields 
\begin{align}
\label{final-upper-no side}
P_{e} &\leq \sum_{i} \sum_{k > i} \sum_{y_v} n^{-\divergence \big( \qtildeij{i}{y_v}, \qtildeij{k}{y_v}\big) + o(1)} .
\end{align}
We now bound the error probability of decoding rule~\eqref{map-detector} from below. Notice that
\begin{align}
\label{noneqiality-lower bound-no side}
\sum_{k\neq i} \Prob\{D \in \bad_{ik} | H_i,y_v\} \leq  (m_{x}-1) \Prob\{D \in \bad_i | H_i,y_v\}. 
\end{align}
Substituting~\eqref{noneqiality-lower bound-no side} into \eqref{error-probability-no side} yields 
\begin{align*}
P_{e} &\geq \frac{1}{m_{x}-1} \sum_i\sum_{k\neq i}\sum_{y_v} \Prob\{D \in \bad_{ik} | H_i,y_v\} P_{i,y_v}  =\frac{1}{m_{x}-1} \\ 
 &\times  \sum_{i}\sum_{k > i} \sum_{y_v} \sum_d \min \Big\{ I(d,i, y_v) P_{i,y_v}, I(d,k,y_v) P_{k,y_v} \Big\}.
\end{align*}
Then it suffices to find a lower bound for~\eqref{main-no side}.
Applying Lemma~\ref{Lemma-Poisson} yields 
\begin{align}
\label{final-lower-no side}
P_{e} &\geq \sum_{i} \sum_{k > i} \sum_{y_v} n^{-\divergence \big(  \qtildeij{i}{y_v},  \qtildeij{k}{y_v} \big) + o(1)} . 
\end{align}
The lower and upper bounds~\eqref{final-upper-no side} and~\eqref{final-lower-no side} imply that the true hypothesis is recovered correctly if $\divergence (  \qtildeij{i}{y_v}, \qtildeij{k}{y_v} ) > 1$ for any $i\neq k$ and any $y_v$. Then under \twomodel{} in which both latent variables $x, y$ are unknown, exact recovery is solvable for $x$ if and only if
\begin{align}
\label{equ: final 2}
\min_{j} \min_{i\neq k} 
\divergence \Big( \qtildeij{i}{j}, \qtildeij{k}{j} \Big) > 1.
\end{align}
%%%%%%%%%%%%%%%%%%%%%%%%%%%%%%%%%%
\section{Proof of Lemma~\ref{Lemma-Poisson-two}}
\label{Proof-Lemma-Poisson-two}
Define
\begin{align*}
    f_{1}(t) &\triangleq \left( \frac{\mathcal{P}_{a}(d) \mathcal{P}_{\hat{a}}(w)}{\mathcal{P}_{b}(d) \mathcal{P}_{\hat{b}}(w)} \right)^{1-t}  ,\\
    f_{2}(t) &\triangleq \left( \frac{\mathcal{P}_{b}(d) \mathcal{P}_{\hat{b}}(w)}{\mathcal{P}_{a}(d) \mathcal{P}_{\hat{a}}(w)} \right)^{t} , \\
    f(t) &\triangleq \mathcal{P}_{a}(d)^{t} \mathcal{P}_{b}(d)^{1-t} \mathcal{P}_{\hat{a}}(w)^{t} \mathcal{P}_{\hat{b}}(w)^{1-t}.
\end{align*}
For any $t\in [0,1]$,
\begin{align}
\label{four-Poisson-upper}
\sum_{d,w \in \mathbb{Z}_{+}^{m}} & \min \{ \mathcal{P}_{a}(d) \mathcal{P}_{\hat{a}}(w) \paaa , \mathcal{P}_{b}(d) \mathcal{P}_{\hat{b}}(w) \pbbb \} \nonumber \\
\leq& \max \{\paaa, \pbbb \} \sum_{d,w \in \mathbb{Z}_{+}^{m}} \min \{ \mathcal{P}_{a}(d) \mathcal{P}_{\hat{a}}(w) , \mathcal{P}_{b}(d) \mathcal{P}_{\hat{b}}(w) \} \nonumber \\
\leq & \max \{\paaa, \pbbb \}  \exp \bigg(-\sum_{i} \big[ta_{i}+(1-t)b_{i}-a_{i}^{t}b_{i}^{1-t} \big] \bigg) \nonumber \\ 
&\times \exp \bigg(-\sum_{i} \big[t\hat{a}_{i}+(1-t)\hat{b}_{i}-\hat{a}_{i}^{t}\hat{b}_{i}^{1-t} \big] \bigg) ,
\end{align}
where the last inequality holds because $\min \{ f_{1}(t), f_{2}(t) \} \leq 1$, and 
\begin{align*}
\sum_{d,w \in \mathbb{Z}_{+}^{m}}  \prod_{i} \frac{ (a_{i}^{t}b_{i}^{1-t})^{d_{i}} }{d_{i}!} e^{-a_{i}^{t}b_{i}^{1-t}}  
\frac{ (\ahat_{i}^{t}\bhat_{i}^{1-t})^{w_{i}} }{w_{i}!} e^{-\ahat_{i}^{t}\bhat_{i}^{1-t}} = 1 .
\end{align*}
For the value of $t$ that minimizes the upper bound of~\eqref{four-Poisson-upper}, we have
\begin{align*}
I(a,b, \ahat, \bhat) &\leq \max \{\paaa, \pbbb \} e^{-\divergence([a,\ahat],[b,\bhat])}.
\end{align*}
Notice that $t^{*}$ satisfies 
\begin{align*}
\prod_{i=1}^{m} \bigg(\frac{b_{i}}{a_{i}}\bigg)^{a_{i}^{t^{*}}b_{i}^{1-t^{*}}} \bigg(\frac{\bhat_{i}}{\ahat_{i}}\bigg)^{\ahat_{i}^{t^{*}}\bhat_{i}^{1-t^{*}}} e^{a_{i}-b_{i}+\ahat_{i}-\bhat_{i}} = 1 .
\end{align*}
Then at the optimal $t^{*}$,
\begin{align*}
\sum_{d,w \in \mathbb{Z}_{+}^{m}} & \min \{ \mathcal{P}_{a}(d) \mathcal{P}_{\hat{a}}(w) \paaa , \mathcal{P}_{b}(d) \mathcal{P}_{\hat{b}}(w) \pbbb \} \nonumber \\
\geq& \min \{\paaa, \pbbb \} \sum_{d,w \in \mathbb{Z}_{+}^{m}} \min \{ \mathcal{P}_{a}(d) \mathcal{P}_{\hat{a}}(w), \mathcal{P}_{b}(d) \mathcal{P}_{\hat{b}}(w) \} \nonumber\\
\overset{(a)}{\geq} & \min \{\paaa, \pbbb \}  e^{-\divergence([a,\ahat],[b,\bhat])} \nonumber \\
&\times \prod_{i} \frac{ (a_{i}^{t^{*}}b_{i}^{1-t^{*}})^{a_{i}^{t^{*}}b_{i}^{1-t^{*}}} }{a_{i}^{t^{*}}b_{i}^{1-t^{*}}!} e^{-a_{i}^{t^{*}}b_{i}^{1-t^{*}}} \nonumber \\
&\times \prod_{i} \frac{ (\ahat_{i}^{t^{*}}\bhat_{i}^{1-t^{*}})^{\ahat_{i}^{t^{*}}\bhat_{i}^{1-t^{*}}} }{\ahat_{i}^{t^{*}}\bhat_{i}^{1-t^{*}}!} e^{-\ahat_{i}^{t^{*}}\bhat_{i}^{1-t^{*}}} \nonumber \\
\overset{(b)}{\geq} & \min \{\paaa, \pbbb \}  e^{-\divergence([a,\ahat],[b,\bhat])} \prod_{i} \frac{1}{e^2} \big [ (a_{i} \ahat_{i} )^{t^{*}} (b_{i}\bhat_{i})^{1-t^{*}} \big ]^{-\frac{1}{2}} ,
\end{align*}
where $(a)$ holds because 
\begin{align*}
\sum_{d,w \in \mathbb{Z}_{+}^{m}} & \min \{ \mathcal{P}_{a}(d) \mathcal{P}_{\hat{a}}(w) , \mathcal{P}_{b}(d) \mathcal{P}_{\hat{b}}(w) \} \\
&\geq  \min \{ \mathcal{P}_{a}(d^*) \mathcal{P}_{\hat{a}}(w^*), \mathcal{P}_{b}(d^*) \mathcal{P}_{\hat{b}}(w^*)\},
\end{align*}
where $d^*$ is defined by $d_i^* \triangleq a_i^{t^*}b_i^{1-t^*}$ and $w^*$ is defined by $w_i^* \triangleq \ahat_i^{t^*}\bhat_i^{1-t^*}$, and $(b)$ is due to Stirling's approximation $n! \leq n^{n+\frac{1}{2}} e^{-n +1}$ for any $n\geq 1$.

%%%%%%%%%%%%%%%%%%%%%%%%%%%%%%%%%%
\section{Proof of Theorem~\ref{theorem-general-cencored}}
\label{proof-theorem-general-cencored}
We aim to recover both $x_{v}$ and $y_{v}$ for node $v$, given a realization of $D$ and a realization of $W$. Our goal is to minimize the error probability by selecting the most likely hypothesis, i.e.,
\begin{align*}
\argmax_{i,j} \Prob\{H_{i,j}| D = d, W =w\}, 
\end{align*}
where 
\begin{align*}
    H_{i,j} : x_v = i, y_v = j .
\end{align*}
The maximum a posteriori (MAP) detector is rewrite as
\begin{align}
\label{MapDetector-general-censored}
\argmax_{i,j} P(d,w| H_{i,j}) P_{i,j} .
\end{align}
Solving~\eqref{MapDetector-general-censored} requires $m_{x}m_y-1$ pairwise comparisons of the hypotheses. From this viewpoint, if 
\begin{align*}
P(d, w| H_{i,j}) P_{i,j} \leq P(d, w| H_{k,l}) P_{k,l}, 
\end{align*}
then a pairwise comparison will choose $H_{k,l}$ over $H_{i,j}$. 
Now assume the correct hypothesis is $H_{i,j}$.
Similar to the proof of Theorems~\ref{theorem-general-side} and~\ref{theorem-general-no side}, it can be shown that the probability of error for recovering the true hypothesis is bounded from above and below by controlling 
\begin{align*}
    \sum_{d,w} \min \{ \mathcal{P}_{\lambdavec{i}{j}}(d) \mathcal{P}_{\lambdavechat{i}{j}}(w) P_{i,j}, \mathcal{P}_{\lambda_{k,l}}(d) \mathcal{P}_{\lambdavechat{k}{l}}(w) P_{k,l} \} .
\end{align*}
It follows from Lemma~\ref{Lemma-Poisson-two} that
\begin{align}
\label{final-upper-general-censored}
P_{e} \leq&  \sum_{i,k > i} \sum_{j, l > j}\max \{P_{i,j}, P_{k,l} \} e^{-\divergence([\lambdavec{i}{j},  \lambdavechat{i}{j}], [\lambdavec{k}{l}, \lambdavechat{k}{l}] )}  \nonumber \\
=&\sum_{i, k > i} \sum_{j, l > j} n^{-\divergence \big( [\qij{i}{j}, \gij{i}{j}], [\qij{k}{l}, \gij{k}{l}]  \big) + o(1)} ,
\end{align}
and
\begin{align}
\label{final-lower-general-censored}
P_{e} \geq&  \sum_{i,k > i} \sum_{j, l > j} \frac{c' \min\{P_{i,j}, P_{k,l} \}}{(\log n)^{m}} e^{-\divergence([\lambdavec{i}{j},  \lambdavechat{i}{j}], [\lambdavec{k}{l}, \lambdavechat{k}{l}] )} \nonumber \\
=&  \sum_{i,k > i} \sum_{j, l > j} n^{-\divergence \big( [\gij{i}{j}, \hij{i}{j}], [\gij{k}{l}, \hij{k}{l}]  \big) + o(1)}, 
\end{align}
where $c'$ is a constant and $m$ is the number of elements in vector $d$, i.e., the product of alphabet sizes of $x_v$ and $y_v$.
The lower and upper bounds~\eqref{final-upper-general-censored} and~\eqref{final-lower-general-censored} imply that the true hypothesis is recovered correctly if $\divergence \big( [\gij{i}{j}, \hij{i}{j}], [\gij{k}{l}, \hij{k}{l}]  \big) > 1$, for any $(i,j) \neq (k,l)$. This means that under the \twocensoredmodel{} all micro-communities are exactly recovered if and only if
\begin{align*}
\min_{(i,j) \neq (k,l)}
\divergence \big( [\gij{i}{j}, \hij{i}{j}], [\gij{k}{l}, \hij{k}{l}]  \big) > 1 .
\end{align*}
% %%%%%%%%%%%%%%%%%%%%%%%%%%%%%%%%%%
\section{Proof of Theorem~\ref{theorem-general-cencored-side}}
\label{proof-theorem-general-cencored-side}

We aim to recover $x_{v}$ when $y_{v}$ is known.
Given a realization of $D$, a realization of $W$, and $y_{v}$, our goal is to minimize the error probability by selecting the most likely hypothesis, i.e.,
\begin{align*}
\argmax_{i} \Prob\{H_i| D = d, W = w,y_v\}, 
\end{align*}
or equivalently,
\begin{align}
\label{MapDetector-SideInformation-censored}
\argmax_{i} \Dist (d| H_i, y_{v}) \Dist (w| H_i, y_{v}) P_{i, y_{v}} ,
\end{align}
which is the MAP detector.
Solving~\eqref{MapDetector-SideInformation-censored} requires $m_{x}-1$ pairwise comparisons of the hypotheses. 
Similar to the proof of Theorem~\ref{theorem-general-side}, it can be shown that the error probability of finding true hypothesis is bounded from above and below by controlling
\begin{align*}
    \sum_{d,w} \min \{ \mathcal{P}_{\lambdavec{i}{y_v}}(d) \mathcal{P}_{\lambdavechat{i}{y_v}}(w) P_{i,y_v}, \mathcal{P}_{\lambda_{k,y_v}}(d) \mathcal{P}_{\lambdavechat{k}{y_v}}(w) P_{k,y_v} \} .
\end{align*}
It follows from Lemma~\ref{Lemma-Poisson-two} that
\begin{align}
\label{final-upper-side-censored}
P_{e} &\leq  \sum_{i} \sum_{k > i} \max \{P_{i,y_v}, P_{k,y_v} \} e^{-\divergence([\lambdavec{i}{y_v},  \lambdavechat{i}{y_v}], [\lambdavec{k}{y_v}, \lambdavechat{k}{l}] )} \nonumber \\
&=\sum_{i} \sum_{k > i} n^{-\divergence \big( [\gij{i}{y_v}, \hij{i}{y_v}], [\gij{k}{y_v}, \hij{k}{y_v}]  \big) + o(1)} . 
\end{align}
and 
\begin{align}
\label{final-lower-side-censored}
P_{e} \geq&  \sum_{i} \sum_{k > i} \frac{c}{(\log n)^{\frac{m}{2}}} e^{-\divergence([\lambdavec{i}{y_v},  \lambdavechat{i}{y_v}], [\lambdavec{k}{y_v}, \lambdavechat{k}{y_v}] )} \nonumber \\
=&  \sum_{i} \sum_{k > i} n^{-\divergence \big( [\gij{i}{y_v}, \hij{i}{y_v}], [\gij{k}{y_v}, \hij{k}{y_v}]  \big) + o(1)}, 
\end{align}
where $c \triangleq c' \min\{P_{i,y_v}, P_{k,y_v} \}$ is a constant and $m$ is the number of elements in vector $d$. The lower and upper bounds~\eqref{final-upper-side-censored} and~\eqref{final-lower-side-censored} imply that the true hypothesis is recovered correctly if $\divergence \big( [\gij{i}{y_v}, \hij{i}{y_v}], [\gij{k}{y_v}, \hij{k}{y_v}]  \big) > 1$, for a given $y_{v}$ and any $i\neq k$. This means that a known latent variable restricts the number of pairwise comparisons. Then under the \twocensoredmodel{} in which the latent variable $y$ is known, and the latent variable $x$ is unknown, exact recovery is possible for $x$ if and only if
\begin{align*}
\min_j \min_{i\neq k} \divergence \big( [\gij{i}{j}, \hij{i}{j}], [\gij{k}{j}, \hij{k}{j}]  \big) > 1.
\end{align*}

% %%%%%%%%%%%%%%%%%%%%%%%%%%%%%%%%%%
\section{Proof of Theorem~\ref{theorem-general-cencored-no side}}
\label{proof-theorem-general-cencored-no side}

We aim to recover $x_{v}$ when $y_{v}$ is unknown, given a realization of $D$ and a realization of $W$ for node $v$.
For this setting the MAP detector is
\begin{align*}
\argmax_{i} \Prob\{H_i| D = d, W=w \}. 
\end{align*} 
For convenience define
\begin{align*}
I(d,w,i, y_v) &\triangleq  \prod_{l} P \bigg(\sum_{j} \dij{l}{j}, \sum_{j} \wij{l}{j} | H_i,y_v \bigg) , 
\end{align*}
where $\sum_{j} \wij{l}{j}$ and $\sum_{j} \dij{l}{j}$ are independent given $H_i$ and $y_v$. 
Then the MAP detector rewrite as 
\begin{align}
\label{map-detector-censored}
\argmax_{i} \sum_{y_v} I(d,w,i, y_v) P_{i,y_v}. 
\end{align}
Solving~\eqref{map-detector-censored} requires $m_{x}-1$ pairwise comparisons. In these comparisons, if
\begin{align*}
\sum_{y_v} I(d,w,i, y_v) P_{i,y_v} < \sum_{y_v} I(d,w,k, y_v) P_{k,y_v} , 
\end{align*}
then we conclude hypothesis $H_i$ is ruled out, i.e., $x_{v} \neq i$, because another hypothesis $H_k$ has a better metric. 
Notice that using the Poisson approximation and the  additive property of Poisson distribution, $I(d,w,i, y_v)$ can be reorganized as
\begin{align*}
I(d,w,i, y_v) =&  \prod_{l} \mathcal{P}_{\sum_{j} \lambdascalar{i}{y_v}{l}{j}} \bigg(\sum_{j} \dij{l}{j} \bigg) \\
&\times \prod_{l} \mathcal{P}_{\sum_{j} \lambdascalar{i}{y_v}{l}{j}} \bigg(\sum_{j} \wij{l}{j} \bigg).
\end{align*}
Similar to the proof of Theorem~\ref{theorem-general-no side}, it can be shown that the error probability of recovering the true hypothesis is bounded from above and below by controlling
\begin{align*}
\sum_{d,w \in \mathbb{Z}_{+}^{m}} \min \Big\{ I(d,w,i, y_v) P_{i,y_v}, I(d,w,k,y_v) P_{k,y_v} \Big\}.
\end{align*}
Applying Lemma~\ref{Lemma-Poisson-two} yields 
\begin{align}
\label{final-upper-no side-censored}
P_{e} &\leq \sum_{i} \sum_{k > i} \sum_{y_v} n^{-\divergence \big( [\gtildeij{i}{y_v}, \htildeij{i}{y_v}], [\gtildeij{k}{y_v}, \htildeij{k}{y_v}] \big) + o(1)} ,
\end{align}
and 
\begin{align}
\label{final-lower-no side-censored}
P_{e} &\geq \sum_{i} \sum_{k > i} \sum_{y_v} n^{-\divergence \big( [\gtildeij{i}{y_v}, \htildeij{i}{y_v}], [\gtildeij{k}{y_v}, \htildeij{k}{y_v}] \big) + o(1)} . 
\end{align}
The lower and upper bounds~\eqref{final-upper-no side-censored} and~\eqref{final-lower-no side-censored} imply that the true hypothesis is recovered correctly if $\divergence \big( [\gtildeij{i}{y_v}, \htildeij{i}{y_v}], [\gtildeij{k}{y_v}, \htildeij{k}{y_v}] \big) > 1$ for any $i\neq k$ and any $y_v$. Then under \twocensoredmodel{} in which both latent variables $x, y$ are unknown, exact recovery is solvable for $x$ if and only if
\begin{align*}
\min_{j} \min_{i\neq k} 
\divergence \big( [\gtildeij{i}{j}, \htildeij{i}{j}], [\gtildeij{k}{j}, \htildeij{k}{j}] \big) > 1.
\end{align*}
%%%%%%%%%%%%%%%%%%%%%%%%%%%%%%%%%%
\section{Proof of Theorem~\ref{Theorem-SDP-oublv}}
\label{Proof-Theorem-SDP-oublv}
We begin by stating sufficient conditions for the optimum solution of~\eqref{oublv-sdp-equ2} matching the true labels $x^*$.

\begin{Lemma}
\label{Lemma_sufficient_oublv}
For the optimization problem~\eqref{oublv-sdp-equ2}, consider the Lagrange multipliers 
\begin{equation*}
\lambda^{*}, \quad D^{*}=\mathrm{diag}(d_{i}^{*}), \quad
S^{*}.
\end{equation*}
If we have 
\begin{align*}
&S^{*} = D^{*}+\lambda^{*}\mathbf{J}-\Tone B -\Ttwo A ,\\
&S^{*} \succeq 0, \\
&\lambda_{2}(S^{*}) >  0 ,\\
&S^{*}x^{*} =0 ,
\end{align*}
then $(\lambda^{*}, D^*, S^*)$ is the dual optimal solution and $\Zsdp=x^{*}x^{*T}$ is the unique primal optimal solution of~\eqref{oublv-sdp-equ2}.
\end{Lemma}
\begin{proof}
Let $D = \text{diag}(d_{i})$, $\lambda \in \mathbb{R}$, and $S \succeq 0$ denote the Lagrangian of \eqref{oublv-sdp-equ2}. For any $Z$ that satisfies the constraints in \eqref{oublv-sdp-equ2}, we have 
\begin{align*}
\Tone  \langle B, Z \rangle +\Ttwo \langle A, Z \rangle \overset{(a)}{\leq}& L(Z,S^{*},D^{*},\lambda^{*}) =\langle D^{*},\mathbf{I}\rangle  \\
\overset{(b)}{=}&\langle S^{*} -\lambda^{*}\mathbf{J}+\Tone B +\Ttwo A, Z^{*} \rangle \\
\overset{(c)}{=}&\Tone  \langle B, Z^{*} \rangle +\Ttwo \langle A, Z^{*} \rangle, 
\end{align*}
where $(a)$ holds because $\langle S^{*},Z \rangle \geq 0$, $(b)$ holds because $Z_{ii}=1$ for all $i \in [n]$ and $S^{*} = D^{*}+\lambda^{*}\mathbf{J}-\Tone B -\Ttwo A$, and $(c)$ holds because $S^{*}x^{*} = \mathbf{0}$ and $x^{*T}\mathbf{1} = 0$. Therefore, $Z^{*}=x^{*}x^{*T}$ is an optimal solution of \eqref{oublv-sdp-equ2}. Now, assume $\tilde{Z}$ is another optimal solution. Then
\begin{align*}
\langle S^{*},  \tilde{Z} \rangle =&\langle D^{*}+\lambda^{*}\mathbf{J}-\Tone B -\Ttwo A ,\tilde{Z} \rangle \\
\overset{(a)}{=}& \langle D^{*}+\lambda^{*}\mathbf{J}-\Tone B -\Ttwo A ,Z^{*} \rangle =\langle S^{*}, Z^{*} \rangle=0 ,
\end{align*}
where $(a)$ holds because $\langle \Tone B +\Ttwo A,Z^{*} \rangle=\langle \Tone B +\Ttwo A,\tilde{Z} \rangle$, $Z_{ii}^{*}=\tilde{Z}_{ii}=1$ for all $i\in [n]$, and $\langle \mathbf{J}, Z^{*} \rangle = \langle \mathbf{J}, \tilde{Z} \rangle = 0$. Since $\tilde{Z} \succeq 0$, and $S^{*}\succeq 0$ while its second smallest eigenvalue $\lambda_{2}(S^{*})$ is positive (since $S^{*} \hat{x}^{*} =\mathbf{0}$), $\tilde{Z}$ must be a multiple of $Z^{*}$. Also, since $\tilde{Z}_{ii}=Z_{ii}^{*}=1$ for all $i \in [n]$, we have $\tilde{Z}=Z^{*}$.
\end{proof}
We now show that $S^{*} = D^{*}+\lambda^{*}\mathbf{J}-\Tone B -\Ttwo A$ satisfies other conditions in Lemma~\ref{Lemma_sufficient_oublv} with probability $1-o(1)$. 
Let 
\begin{equation} 
d_{i}^{*} = \Tone  \sum_{j=1}^{n} B_{ij}x_{j}^{*}x_{i}^{*} +\Ttwo  \sum_{j=1}^{n} A_{ij}x_{j}^{*}x_{i}^{*} .
\end{equation}
Then $D^{*}x^{*}=\Tone Bx^{*} +\Ttwo Ax^{*}$ and based on the definition of $S^{*}$ in  Lemma~\ref{Lemma_sufficient_oublv}, $S^{*}$ satisfies the condition $S^{*}x^{*} =0$.
It remains to show that $S^{*}\succeq 0$ and $\lambda_{2}(S^{*})>0$ with probability $1-o(1)$. In other words, we need to show that
\begin{equation}
\label{equ-main-sdp}
\Prob  \bigg\{ \underset{v\perp x^{*},  \| v   \|=1}{\inf} v^{T}S^{*}v>0  \bigg \}\geq 1-o(1) ,
\end{equation}
where $v$ is a $n \times 1$ vector. 
Then for any $v$ such that $v^{T}x^{*}=0$ and $  \| v   \|=1$,
\begin{align*}
v^{T}S^{*}v =&v^{T}D^{*}v +\lambda^{*}v^{T}\mathbf{J}v -\Tone v^{T} (B-\mathbb{E} [B ] )v \twocolbreak -\Ttwo  v^{T} (A-\mathbb{E} [A ] )v -\Tone v^{T}\mathbb{E} [B ]v -\Ttwo  v^{T}\mathbb{E} [A ]v\\
\geq& \min_{i}d_{i}^{*} +\lambda^{*}v^{T}\mathbf{J}v -\Tone   \| B-\mathbb{E} [B ]   \| \twocolbreak-\Ttwo    \| A-\mathbb{E} [A ]   \| -\Tone v^{T}\mathbb{E} [B ]v -\Ttwo  v^{T}\mathbb{E} [A ]v .
\end{align*}
Notice that
\begin{align*}
\Tone v^{T}\mathbb{E} [B ]v  +\Ttwo  v^{T}\mathbb{E} [A ]v =& \frac{1}{4}  [\Tone c_{1} +\Ttwo c_{2}  ]v^{T}Wv \\
&+\frac{1}{4}  [\Tone c_{3} +\Ttwo c_{4}  ]v^{T}(Z*W)v \\
&+\frac{1}{4}  [\Tone c_{1} +\Ttwo c_{2}  ]v^{T}\mathbf{J}v \\
&- (\Tone +\Ttwo  ) \qzero{} \frac{\log n}{n} ,
\end{align*}
where 
\begin{align*}
c_{1} &\triangleq  \frac{\log n}{n}( \qzero{} -\qone{} +\qtwo{} -\qthree{}), \\ 
c_{2} &\triangleq  \frac{\log n}{n} (\qzero{} +\qone{} +\qtwo{} +\qthree{}) , \\
c_{3} &\triangleq \frac{\log n}{n} (\qzero{} -\qone{} -\qtwo{} +\qthree{}) , \\
c_{4} &\triangleq \frac{\log n}{n} (\qzero{} +\qone{} -\qtwo{} +\qthree{}) .
\end{align*}

\begin{Lemma}
\label{Lemma_spectral_norm_1}
For any $c > 0$, there exists $c' , c''>0$ such that for any $n \geq 1$, $  \| A-\mathbb E[A]   \| \leq c''\sqrt{\log n}$ and $  \| B-\mathbb E[B]   \| \leq c'\sqrt{\log n}$ with probability at least $1-n^{-c}$. 
\end{Lemma}
\begin{proof}
The proof is similar to the proofs \cite[Thoerem 9]{hajek2016achievingExtensions} and \cite[Thoerem 5]{hajek2016achieving}. 
\end{proof}

\begin{Lemma}
\label{Lemma-vWv}
With probability at least $1-n^{-\frac{1}{2}}$,
\begin{align*}
&v^{T}(Z*W)v \leq \sqrt{\log n} , \\
&v^{T}Wv \leq \sqrt{\log n}+ (2\rhoy-1 )^{2} v^{T}\mathbf{J}v + 2|2\rhoy-1|\sqrt{n \log n} . 
\end{align*}
\end{Lemma}
\begin{proof}
Since $-   | v_{i}   | \leq v_{i}y_{i}\leq   | v_{i}   |$, by applying the Chernoff bound we have
\begin{align*}
    \Prob  (v^{T}y-\mathbb{E} [v^{T}y] \geq \sqrt{\log n} ) \leq n^{-\frac{1}{2}}.
\end{align*}
Since $\mathbb{E} [v^{T}y ] =  (2\rhoy-1 ) v^{T}\mathbf{1}$ and $| v^{T}\mathbf{1} | \leq \| v \|_{2} \| \mathbf{1}\|_{2} = \sqrt{n}$, with probability converging to one,
\begin{align*}
 (v^{T}y )^{2} \leq& \log n +  (2\rhoy-1 )^{2} v^{T}\mathbf{J}v  +2| v^{T}\mathbf{1} |  | 2\rhoy-1|\sqrt{\log n} \\
\leq& \log n +  (2\rhoy-1 )^{2} v^{T}\mathbf{J}v +2 | 2\rhoy-1|\sqrt{n \log n} . 
\end{align*}
Similarly, since $\mathbb{E} [ \sum_{i} x_{i}y_{i}v_{i} ] = 0$ and $-   | v_{i}   | \leq x_{i}y_{i}v_{i}\leq   | v_{i}   |$, applying the Chernoff bound yields $v^{T}(Z*W)v \leq \sqrt{\log n}$ with probability converging to one.
\end{proof}

\begin{Lemma}
\label{Lemma_main_one}
For $\delta = \frac{\log n}{ \log \log n}$,
\begin{equation*}
    \Prob \Big(\min_{i\in [n]}~d_{i}^{*} \geq \delta  \Big ) \geq 1-n^{1-\eta_{1}(\qbf,\rho)+o(1)} -n^{1-\eta_{1}(\qbf,1-\rho)+o(1)} .
\end{equation*}
\end{Lemma}
\begin{proof}
The proof is achieved by applying the Chernoff bound and taking the union bound. 
\end{proof}

Notice that $\rhoy \leq 0.5$ implies $\eta_{1}(\qbf,\rho) \leq \eta_{1}(\qbf,1-\rho)$ and $\rhoy > 0.5$ implies $\eta_{1}(\qbf,\rho) \geq \eta_{1}(\qbf,1-\rho)$. Then $\min_{i}d_{i}^{*} \geq \frac{\log n}{\log\log n}$ if 
\begin{align}
\label{equ-theorem3-final-conditions}
    \begin{cases}
    \eta_{1}(\qbf,\rho)> 1 & \text{when} \quad \rhoy \leq 0.5 \\
    \eta_{1}(\qbf,1-\rho)> 1 & \text{when} \quad \rhoy > 0.5
    \end{cases} .
\end{align}
Let $\lambda^{*} \geq \frac{1}{4}  [\Tone c_{1} +\Ttwo c_{2}  ]  (2\rhoy-1 )^{2}$. Therefore, applying Lemmas \ref{Lemma_spectral_norm_1}, \ref{Lemma-vWv}, and \ref{Lemma_main_one}, we get that if~\eqref{equ-theorem3-final-conditions} holds, then
\begin{align*}
v^{T}S^{*}v \geq& \frac{\log n}{\log \log n} - ( \Tone c'+\Ttwo c''  ) \sqrt{\log n}  \\
&+ (\Tone +\Ttwo  ) \qzero{} \frac{\log n}{n} > 0 ,
\end{align*}
and the first part of Theorem~\ref{Theorem-SDP-oublv} follows.

To prove the second part, 
since $x^{*}$ has a uniform distribution over $\{x\in \{\pm 1\}^{n}: x^{T}\mathbf{1} = 0 \}$, maximum likelihood estimator minimizes the error probability among all estimators. Then we need to find when the maximum likelihood estimator fails. Let $e(i,\Hset) \triangleq \sum_{j \in \Hset} A_{ij}(\Tone y_{i}y_{j}+\Ttwo )$. 
Define the events
\begin{align*}
    F_{1} &\triangleq  \Big \{ \min_{i \in C_{1}^{*}}   ( e(i,C_{1}^{*})-e(i,C_{2}^{*})   ) \leq -2 \Big \} ,\\
    F_{2} &\triangleq  \Big \{\min_{i \in C_{2}^{*}}   ( e(i,C_{2}^{*})-e(i,C_{1}^{*})   ) \leq -2 \Big \} ,
\end{align*}
where $C_{1}^{*} = \{v\in [n]: x_{v}^{*} = 1\}$ and $C_{2}^{*} = \{v\in [n]: x_{v}^{*} = -1\}$. Then $\Prob  ( \text{ML fails}   ) \geq \Prob  ( F_{1} \cap F_{2}   )$. Thus, it suffices to show that with high probability $\Prob (F_{1} )  \rightarrow 1$ and $\Prob (F_{2} )  \rightarrow 1$. Here we just prove that $\Prob (F_{1} )  \rightarrow 1$, while $\Prob (F_{2} )  \rightarrow 1$ is proved similarly. By symmetry, we can condition on $C_{1}^{*}$ being the first $\frac{n}{2}$ nodes. Let $\Tset$ denote the set of first $  \lfloor \frac{n}{\log^{2} n}   \rfloor$ nodes of $C_{1}^{*}$. Then
\begin{align*}
\min_{i \in C_{1}^{*}}   ( e(i,C_{1}^{*})-e(i,C_{2}^{*})   ) \leq& \min_{i \in \Tset}   ( e(i,C_{1}^{*})-e(i,C_{2}^{*})   ) \\
\leq& \min_{i \in \Tset}   ( e(i,C_{1}^{*} \setminus \Tset)-e(i,C_{2}^{*})   ) \twocolbreak+\max_{i \in \Tset} e(i, \Tset) .
\end{align*}
Define the events
\begin{align*}
    E_{1} &\triangleq \Big \{\max_{i \in \Tset} e(i, \Tset) \leq \delta -2 \Big \}, \\
    E_{2} &\triangleq \Big \{\min_{i \in \Tset}   ( e(i,C_{1}^{*} \setminus \Tset)-e(i,C_{2}^{*})   ) \leq -\delta \Big \} . 
\end{align*}
It suffices to show that $\Prob (E_{1} )  \rightarrow 1$ and $\Prob (E_{2} )  \rightarrow 1$, to have $\Prob (F_{1} )  \rightarrow 1$.
For any $i \in \Tset$, 
\begin{align*}
    e(i,\Tset) = (\Ttwo +\Tone ) X_{1} + (\Ttwo -\Tone ) X_{2}, 
\end{align*}
where $X_{1} \sim \text{Binom}(|\Tset|,\qzero \log n/n)$ and $X_{2} \sim \text{Binom}(|\Tset|,\qone\log n/n)$. 

\begin{Lemma}\cite[Lemma 5]{esmaeili2019community}
\label{BCBM-P-Lemma 5}
When $S \sim \Bin(n,p)$, for any $r\geq 1$, 
\begin{align*}
    \Prob( S \geq rnp  ) \leq  \Big(\frac{e}{r} \Big)^{rnp} e^{-np} .
\end{align*}
\end{Lemma}

From Lemma~\ref{BCBM-P-Lemma 5},
\begin{align*}
\Prob  \bigg( X_{1} \geq \frac{\delta -2}{2(\Tone +\Ttwo )}  \bigg) &\leq  \bigg( \frac{(\delta-2)\log n}{4(\Tone +\Ttwo )e\qzero}  \bigg)^{\frac{2-\delta }{2(\Tone +\Ttwo )}} \twocolbreak \leq  n^{-2+o(1)} , \\
\Prob  \bigg( X_{2} \geq \frac{\delta -2}{2  | \Ttwo -\Tone    | }  \bigg) &\leq  \bigg( \frac{(\delta-2)\log n}{4  | \Ttwo -\Tone    |e\qone}  \bigg)^{\frac{2-\delta }{2  | \Ttwo -\Tone    | }} \twocolbreak \leq  n^{-2+o(1)} .
\end{align*}
Since $|\Ttwo -\Tone|  > 0$ and $\Tone +\Ttwo  > 0$, 
\begin{align*}
\Prob & ( e(i,\Tset) \geq \delta-2  ) \\
&\leq  \Prob  ( (\Tone +\Ttwo ) X_{1} +  | \Ttwo -\Tone  | X_{2} \geq \delta-2  ) \leq n^{-2+o(1)} .
\end{align*}
Using the union bound yields $\Prob (E_{1} ) \geq 1-n^{-1+o(1)}$. Therefore, $\Prob (E_{1} )  \rightarrow 1$ with high probability. 

\begin{Lemma}~\cite[Lemma 15]{saad2018community}
\label{BASBM-N-Lemma 5}
Let $\{S_1,\ldots,S_m\}$ be a sequence of i.i.d. random variables, where $m-n =o(n)$. Then for any $\mu \in \mathbb{R}$ and $\nu \geq 0$ we have
\begin{align*}
\Prob \Bigg( \sum_{i=1}^{m} S_{i} \geq \mu-\nu  \Bigg) \geq \min_{t >0}~ e^{ -t\mu-  | t   |\nu } M(t)  \bigg( 1-\frac{\sigma_{\hat{Z}}^{2}}{\nu^{2}}  \bigg) ,
\end{align*}
where $M(t)$ is the moment generating function of $Z=\sum_{i=1}^{m} S_{i}$ and $\hat{Z}$ is a random variable distributed according to $\frac{e^{tz}\Prob(z)}{E_{Z} [e^{tz} ]}$ with variance $\sigma_{\hat{Z}}^{2}$.
\end{Lemma}

\begin{Lemma}
\label{Lemma_converse_prob_1}
Let $e(i,\Hset) \triangleq \sum_{j \in \Hset} A_{ij}(\Tone y_{i}y_{j}+\Ttwo )$. 
Define 
\begin{align*}
    E_{2}' \triangleq \big \{ e(i,C_{1}^{*} \setminus \Tset)-e(i,C_{2}^{*}) \leq -\delta \big\}. 
\end{align*}
Then
\begin{align*}
    \Prob  ( E_{2}'   ) \geq n^{-\eta_{1}(\qbf,\rho)+o(1)} + n^{-\eta_{1}(\qbf,1-\rho)+o(1)}. 
\end{align*}
\end{Lemma}
\begin{proof}
The proof is achieved by applying Lemma~\ref{BASBM-N-Lemma 5} and the Chernoff bound. 
\end{proof}
Applying Lemma~\ref{Lemma_converse_prob_1} yields
\begin{align*}
\Prob  (E_{2} ) &=1-\prod_{i\in \Tset}   [ 1-\Prob  ( E_{2}'   )   ] \\
&\geq 1-  \Big [ 1-n^{-\eta_{1}(\qbf,1-\rho)+o(1)} -n^{-\eta_{1}(\qbf,\rho)+o(1)}   \Big ]^{|\Tset|} \\
&\geq 1-e^{ -n^{1-\eta_{1}(\qbf,1-\rho)+o(1)} -n^{1-\eta_{1}(\qbf,\rho)+o(1)} } .
\end{align*}
Recall that $\rhoy \leq 0.5$ implies $\eta_{1}(\qbf,\rho) \leq \eta_{1}(\qbf,1-\rho)$ and $\rhoy > 0.5$ implies $\eta_{1}(\qbf,\rho) \geq \eta_{1}(\qbf,1-\rho)$.  When $\rhoy \leq 0.5$, if $\eta_{1}(\qbf,\rho) <1$ then $\Prob  (E_{2} )  \rightarrow 1$. When $\rhoy \geq 0.5$, if $\eta_{1}(\qbf,1-\rho) <1$ then $\Prob  (E_{2} )  \rightarrow 1$ and the second part of Theorem~\ref{Theorem-SDP-oublv} follows.

%%%%%%%%%%%%%%%%%%%%%%%%%%%%%%%%%%
\section{Partial Recovery Algorithm}
\label{Partial Recavery Algorithm}
In this paper, the partial recovery algorithm in~\cite{abbe2015community} is employed with few changes to make it compatible for each Scenario.
For the two-latent variable stochastic block model we can directly use the partial recovery algorithm in~\cite{abbe2015community}:
\subsection{The two-latent variable stochastic block model with known auxiliary latent variable $y$:}
\begin{enumerate}
\item
Cluster nodes according to the value of the auxiliary latent variable $y$, call them auxiliary clusters.
\item
Extract submatrices of $P$ and $\bar{Q}$ representing each value of $y$, call them $P^{(k)}$ and $\bar{Q}^{(k)}$.
\item
Separately in each auxiliary cluster, use respective submatrices $P^{(k)}$ and $\bar{Q}^{(k)}$ to construct a partial recovery estimator of communities $x$, and find the community estimate for all members of each cluster.
\end{enumerate}

% \begin{enumerate}
% \item Clustering nodes with the same latent variable $y_i=k$. 
% \item Constructing matrices $P^{(k)}$ and $\bar{Q}^{(k)}$ (corresponding to the original matrices $P$ and $Q$) for each set of nodes with the same latent variable $y_i=k$.
%     \item Finding the partial recovery estimator to cluster each set of nodes with the same latent variable $y_i=k$ into $m_x$ communities via matrices $P^{(k)}$ and $\bar{Q}^{(k)}$.
%     \item Clustering nodes with the same community latent variable $x_i$ predicted by the estimator.
% \end{enumerate}
    
\subsection{The two-latent variable stochastic block model with unknown latent variable $y$:}
\begin{enumerate}
    \item Use matrices $P$ and $\bar{Q}$ to construct a partial recovery estimator of all micro-communities. 
    \item Cluster nodes with the same community variable representing each value of $x$.
\end{enumerate}
% \begin{enumerate}
%     \item Finding the partial recovery estimator to partially  recover $m_x m_y$ micro-communities, using the original matrices $P$ and $\bar{Q}$. 
%     \item Clustering nodes with the same community latent variable $x_i$ predicted by the estimator.
% \end{enumerate}
For the two-latent variable censored block model, we need a new variant of the partial recovery algorithm in~\cite{abbe2015community}. In the new variant, the \textit{vertex comparison algorithm} in \cite{abbe2015community} is used twice for each pair of nodes. First, the algorithm is employed using the eigenvalues of $\text{diag} (p) (\Xi * Q)$. For this case, if the two nodes belong to the same community, the output of the algorithm is 1; otherwise it returns 0. Then, the algorithm is employed using the eigenvalues of $\text{diag} (p) ((1-\Xi) * Q)$. For this case, if the two nodes belong to the same community, the output of the algorithm is 0; otherwise it returns 1. If the outputs are not equal, we are able to determine reliably whether the two nodes belong to the same community. If the outputs are equal, another pair of nodes are selected to repeat the partial recovery algorithm. 

\subsection{The two-latent variable censored block model with known latent variable $y$:}
\begin{enumerate}
\item
Cluster nodes according to the value of the auxiliary latent variable $y$, call them auxiliary clusters.
\item
Extract submatrices of $P$, $\bar{Q}$, and $\Xi$ representing each value of $y$, call them $P^{(k)}$ and $\bar{Q}^{(k)}$, and $\Xi^{(k)}$.
\item
Separately in each auxiliary cluster, use respective submatrices $P^{(k)}$, $\bar{Q}^{(k)}$, and $\Xi^{(k)}$ to construct a partial recovery estimator of communities $x$, and find the community estimate for all members of each cluster.
\end{enumerate}

% \begin{enumerate}
% \item Clustering nodes with the same latent variable $y_i=k$. 
% \item  Constructing matrices $P^{(k)}$, $\bar{Q}^{(k)}$, and $\Xi^{(k)}$ (corresponding to the original matrices $P$, $Q$, and $\Xi$) for each set of nodes with the same latent variable $y_i=k$.
%     \item Finding the partial recovery estimator to cluster each set of nodes with the same latent variable $y_i=k$ into $m_x$ communities via considering matrices $P^{(k)}$, $\bar{Q}^{(k)}$, and $\Xi^{(k)}$.
%     \item Clustering nodes with the same community latent variable $x_i$ predicted by the estimator.
% \end{enumerate}

\subsection{The two-latent variable censored block model with unknown latent variable $y$:}
\begin{enumerate}
    \item Use matrices $P$, $\bar{Q}$, and $\Xi$ to construct a partial recovery estimator of all micro-communities. 
    \item Cluster nodes with the same community variable representing each value of $x$.
\end{enumerate}

\begin{Remark}
When $y$ is known, for each auxiliary latent variable $y$, 
definitions 4 and 5 in~\cite{abbe2015community} are restated based on the new matrices $P^{(k)}$, $\bar{Q}^{(k)}$, and $\Xi^{(k)}$. Using these new matrices, \textbf{the vertex comparison algorithm}, \textbf{the vertex classification algorithm}, \textbf{the unreliable graph classification algorithm}, and \textbf{the reliable graph classification algorithm} in~\cite{abbe2015community} are exploited separately. When $y$ is unknown, these definitions and algorithms are followed from matrices $P$, $\bar{Q}$, and $\Xi$.
\end{Remark}
% \begin{enumerate}
%     \item Finding the partial recovery estimator to partially recover $m_x m_y$ micro-communities,  considering the original matrices $P$, $\bar{Q}$, and $\Xi$.
%     \item Clustering nodes with the same community latent variable $x_i$ predicted by the estimator.
% \end{enumerate}
%%%%%%%%%%%%%%%%%%%%%%%%%%%%%%%%%%
\section{Proof of Theorem \ref{Theorem-SDP-tublv}}
\label{Proof-Theorem-SDP-tublv}
We begin by deriving sufficient conditions for the semidefinite programming estimator~\eqref{tublv-sdp} to produce the true labels $x^*$.

\begin{Lemma}
\label{Lemma_sufficient_tublv}
The sufficient conditions of Lemma~\ref{Lemma_sufficient_oublv} apply to semidefinite programming \eqref{tublv-sdp} by replacing
\begin{align*}
    S^{*} = D^{*}+\lambda^{*}\mathbf{J}-A .
\end{align*}
\end{Lemma}
\begin{proof}
The proof is similar to the proof of Lemma~\ref{Lemma_sufficient_oublv}.
\end{proof}

It suffices to show that $S^{*} = D^{*}+\lambda^{*}\mathbf{J}-A$ satisfies other conditions in Lemma~\ref{Lemma_sufficient_tublv} with probability $1-o(1)$. 
Let 
\begin{equation*} 
d_{i}^{*} = \sum_{j=1}^{n} A_{ij}x_{j}^{*}x_{i}^{*}.
\end{equation*}
Then $D^{*}x^{*}=Ax^{*}$ and based on the definition of $S^{*}$ in  Lemma~\ref{Lemma_sufficient_tublv}, $S^{*}$ satisfies the condition $S^{*}x^{*} =0$.
It remains to show that $S^{*}\succeq 0$ and $\lambda_{2}(S^{*})>0$ with probability $1-o(1)$, i.e., ~\eqref{equ-main-sdp} holds. 
For any $v$ such that $v^{T}x^{*}=0$ and $  \| v   \|=1$,
\begin{align*} 
v^{T}S^{*}v =&v^{T}D^{*}v +\lambda^{*}v^{T}\mathbf{J}v -v^{T} (A-\mathbb{E} [A ] )v -v^{T}\mathbb{E} [A ]v \\
\geq& \min_{i}d_{i}^{*} +\lambda^{*}v^{T}\mathbf{J}v -  \| A-\mathbb{E} [A ]   \| -v^{T}\mathbb{E} [A ]v .
\end{align*}
Notice that
\begin{align*}
    -v^{T}\mathbb{E} [A ]v =& -\frac{1}{4}  [c_{1} v^{T}Wv - c_{2} v^{T}\mathbf{J}v - c_{3} v^{T}(Z*W)v  ] \twocolbreak+ \qzero{} \frac{\log n}{n} .
\end{align*}
 
\begin{Lemma}
\label{Lemma_main_two}
For $\delta = \frac{\log n}{ \log \log n}$,  
\begin{align*}
    \Prob \big ( \min_{i}~d_{i}^{*} \geq \delta  \big ) \geq 1-n^{1-\eta_{2}(\qbf,\rho)+o(1)} -n^{1-\eta_{2}(\qbf,1-\rho)+o(1)} . 
\end{align*}
\end{Lemma}
\begin{proof}
The proof is achieved by applying the Chernoff bound and the union bound. 
\end{proof}

Using Lemma~\ref{Lemma_main_two}, $\min_{i}d_{i}^{*} \geq \frac{\log n}{\log\log n}$ with probability converging to one, if $\min \{\eta_{2}(\qbf,\rho), \eta_{2}(\qbf,1-\rho)\}> 1$. Let $\lambda^{*} \geq \frac{1}{4}  [c_{1}  (2\rhoy-1 )^{2} +c_{2}  ]$. Applying Lemmas \ref{Lemma_spectral_norm_1}, \ref{Lemma-vWv}, and \ref{Lemma_main_two}, we get that when $\min \{\eta_{2}(\qbf,\rho), \eta_{2}(\qbf,1-\rho) \}> 1$,
\begin{align*}
v^{T}S^{*}v \geq& \frac{\log n}{\log \log n} -c' \sqrt{\log n} +\qzero{} \frac{\log n}{n} > 0 ,
\end{align*}
and the first part of Theorem~\ref{Theorem-SDP-tublv} follows.

%%%%%%%%%%%%%%%%%%%%%%%%%%%%%%%%%%%%%%%%%%%%%%%%%%%%%%%%%%%%%%%%%%%%%%%%%%%%%%%%%%%%%%%%%%%%%%%%%%%%%%%%%%%%%%%%%%%%%%%%%%%%%%%%%%%%
To prove the second part, it suffices to find when the maximum likelihood detector fails. 
The events $F_{1}$, $F_{2}$, $E_{1}$, $E_{2}$, and $E_{2}'$ are the same as we defined them in the proof of Theorem~\ref{Theorem-SDP-oublv}. Also, the definitions for $C_{1}^{*}$, $C_{2}^{*}$, and $\Tset$ remain valid for this part. Then $\Prob  ( \text{ML fails}   ) \geq \Prob  ( F_{1} \cap F_{2}   )$. Here we just prove that $\Prob (F_{1} )  \rightarrow 1$, while $\Prob (F_{2} )  \rightarrow 1$ is proved similarly. 
By symmetry, we can condition on $C_{1}^{*}$ being the first $\frac{n}{2}$ nodes. Then
\begin{align*}
\min_{i \in C_{1}^{*}}   ( e(i,C_{1}^{*})-e(i,C_{2}^{*})   ) \leq& \min_{i \in \Tset}   ( e(i,C_{1}^{*})-e(i,C_{2}^{*})   ) \\
\leq & \min_{i \in \Tset}   ( e(i,C_{1}^{*} \setminus \Tset)-e(i,C_{2}^{*})   ) \twocolbreak+\max_{i \in \Tset} e(i, \Tset) ,
\end{align*}
where $e(i,\Hset) \triangleq \sum_{j \in \Hset} A_{ij}$. For $i \in \Tset$, $e(i,\Tset) = X_{1} + X_{2}$, where $X_{1} \sim \text{Binom}(|\Tset|,\qzero \log n/n)$ and $X_{2} \sim \text{Binom}(|\Tset|,\qone \log n/n)$. It follows from Lemma~\ref{BCBM-P-Lemma 5} that
\begin{align*}
&\Prob  \bigg( X_{1} \geq \frac{\delta}{2}-1  \bigg) \leq  \bigg( \frac{\log n}{2e\qzero}  \Big(\frac{\delta}{2}-1 \Big)  \bigg)^{1-\frac{\delta }{2}} \leq  n^{-2+o(1)} ,
\end{align*}
\begin{align*}
&\Prob  \bigg( X_{2} \geq \frac{\delta}{2}-1  \bigg) \leq  \bigg( \frac{\log n}{2e\qone}  \Big(\frac{\delta}{2}-1 \Big)  \bigg)^{1-\frac{\delta }{2}} \leq  n^{-2+o(1)} .
\end{align*}
Then $\Prob  ( e(i,\Tset) \geq \delta-2  ) \leq n^{-2+o(1)}$. 
Using the union bound, $\Prob (E_{1} ) \geq 1-n^{-1+o(1)}$. Therefore, $\Prob (E_{1} )  \rightarrow 1$ with high probability. 

\begin{Lemma}
\label{Lemma_converse_prob_2}
When $e(i,\Hset) \triangleq \sum_{j \in \Hset} A_{ij}$, 
\begin{align*}
    \Prob  ( E_{2}'   ) \geq n^{-\eta_{2}(\qbf, \rho)+o(1)} + n^{-\eta_{2}(\qbf, 1-\rho)+o(1)} . 
\end{align*}
\end{Lemma}
\begin{proof}
The proof is achieved by applying Lemma~\ref{BASBM-N-Lemma 5} and the Chernoff bound. 
\end{proof}
Applying Lemma~\ref{Lemma_converse_prob_2} yields
\begin{align*}
\Prob  (E_{2} ) &=1-\prod_{i\in \Tset}   [ 1-\Prob  ( E_{2}'   )   ] \\
&\geq 1-  \Big [ 1-n^{-\eta_{2}(\qbf,\rho)+o(1)} -n^{-\eta_{2}(\qbf,1-\rho)+o(1)}    \Big ]^{|\Tset|} \\
&\geq 1-e^{ -n^{1-\eta_{2}(\qbf,\rho)+o(1)} -n^{1-\eta_{2}(\qbf,1-\rho)+o(1)} } .
\end{align*}
Therefore, if $\min \{\eta_{2}(\qbf,\rho), \eta_{2}(\qbf,1-\rho)\} <1$ then $\Prob  (E_{2} )  \rightarrow 1$ and the second part of Theorem~\ref{Theorem-SDP-tublv} follows.
%%%%%%%%%%%%%%%%%%%%%%%%%%%%%%%%%%
\section{Proof of Theorem \ref{Theorem-SDP-oublv-censored}}
\label{Proof-Theorem-SDP-oublv-censored}
The proof is similar to the proof of Theorem~\ref{Theorem-SDP-oublv}. Here we just mention the proof outlines and important Lemmas for brevity.
The following Lemma declares the sufficient conditions for the optimum solution of~\eqref{oublv-sdp-equ2-censored} matching the true labels $x^*$.
\begin{Lemma}
\label{Lemma_sufficient_oublv-censored}
For the optimization problem~\eqref{oublv-sdp-equ2-censored}, consider the Lagrange multipliers 
\begin{equation*}
\lambda^{*}, \quad D^{*}=\mathrm{diag}(d_{i}^{*}), \quad
S^{*}.
\end{equation*}
If we have 
\begin{align*}
&S^{*} = D^{*}+\lambda^{*}\mathbf{J}-R ,\\
&S^{*} \succeq 0, \\
&\lambda_{2}(S^{*}) >  0 ,\\
&S^{*}x^{*} =0 ,
\end{align*}
then $(\lambda^{*}, D^*, S^*)$ is the dual optimal solution and $\Zsdp=x^{*}x^{*T}$ is the unique primal optimal solution of~\eqref{oublv-sdp-equ2-censored}.
\end{Lemma}
\begin{proof}
The proof is similar to the proof of Lemma~\ref{Lemma_sufficient_oublv}.
\end{proof}
Let
\begin{align*} 
d_{i}^{*} =& T \sum_{j=1}^{n} A_{ij} x_{j}^{*}x_{i}^{*} +T \sum_{j=1}^{n} A_{ij}y_i y_j x_{j}^{*}x_{i}^{*} \\
&+\Tone  \sum_{j=1}^{n} A_{ij}^2 y_i y_j x_{j}^{*}x_{i}^{*} +\Ttwo \sum_{j=1}^{n} A_{ij}^2 x_{j}^{*}x_{i}^{*}.
\end{align*}
Then $D^{*}x^{*}=T A +T (A*W)  +\Tone  (A*A*W) +\Ttwo (A*A)$ and based on the definition of $S^{*}$ in  Lemma~\ref{Lemma_sufficient_oublv-censored}, $S^{*}$ satisfies the condition $S^{*}x^{*} =0$.
\begin{Lemma}
\label{Lemma_main_one-censored}
For $\delta = \frac{\log n}{ \log \log n}$,
\begin{align*}
    \Prob \Big( \min_{i\in [n]}~d_{i}^{*} \geq \delta  \Big ) \geq& 1-n^{1-\eta_{1}(\gbf,\rho) -\eta_{1}(\hbf,\rho) +o(1)} \\
    &-n^{1-\eta_{1}(\gbf,1-\rho) -\eta_{1}(\hbf,1-\rho) +o(1)} .
\end{align*}
\end{Lemma}
\begin{proof}
The proof is achieved by applying the Chernoff bound and taking the union bound. 
\end{proof}
Similar to the proof of Theorem~\ref{Theorem-SDP-oublv}, using Lemma~\ref{Lemma_main_one-censored}, it can be shown that
 $S^{*}\succeq 0$ and $\lambda_{2}(S^{*})>0$ with probability $1-o(1)$ if
\begin{align*}
\begin{cases}
\eta_1(\gbf,\rhoy) +\eta_1(\hbf,\rhoy) >1 & \text{when} \quad \rhoy \leq 0.5 \\
\eta_1(\gbf,1-\rhoy) +\eta_1(\hbf,1-\rhoy) >1 & \text{when} \quad \rhoy > 0.5
\end{cases} .
\end{align*}

To prove the second part, we start to find when the maximum likelihood estimator fails. To this end, let 
\begin{align*}
    e(i,\Hset) \triangleq \sum_{j \in \Hset} A_{ij} (T y_{i}y_{j}+T)  + A_{ij}^2 (\Tone  y_{i}y_{j} +\Ttwo ).
\end{align*}
The definition of events $F_1$, $F_2$, $E_1$, and $E_2$ in the proof of Theorem~\ref{Theorem-SDP-oublv} are used to show that with high probability $\Prob (F_{1} )  \rightarrow 1$ and $\Prob (F_{2} )  \rightarrow 1$. Also, the definitions for $C_{1}^{*}$, $C_{2}^{*}$, and $\Tset$ remain valid for this part.
We prove that $\Prob (F_{1} )  \rightarrow 1$, while $\Prob (F_{2} )  \rightarrow 1$ is proved similarly.
To show that $\Prob (F_{1} )  \rightarrow 1$, we must have
$\Prob (E_{1} )  \rightarrow 1$ and $\Prob (E_{2} )  \rightarrow 1$. 
It can be shown that $\Prob (E_{1} ) \geq 1-n^{-1+o(1)}$ without difficulty. 
\begin{Lemma}
\label{Lemma_converse_prob_1-censored}
Let $E_{2}' \triangleq \big \{ e(i,C_{1}^{*} \setminus \Tset)-e(i,C_{2}^{*}) \leq -\delta \big\}$.
Then
\begin{align*}
    \Prob  ( E_{2}'   ) \geq& n^{-\eta_{1}(\gbf,\rho) -\eta_{1}(\hbf,\rho) +o(1)} \\
    &+ n^{-\eta_{1}(\gbf,1-\rho) -\eta_{1}(\hbf,1-\rho) +o(1)}. 
\end{align*}
\end{Lemma}
\begin{proof}
The proof is achieved by applying Lemma~\ref{BASBM-N-Lemma 5} and the Chernoff bound. 
\end{proof}
Applying Lemma~\ref{Lemma_converse_prob_1-censored} yields
\begin{align*}
\Prob  (E_{2} ) &=1-\prod_{i\in \Tset}   [ 1-\Prob  ( E_{2}'   )   ] \\
&\geq 1-e^{ -n^{-\eta_{1}(\gbf,\rho) -\eta_{1}(\hbf,\rho) +o(1)} -n^{-\eta_{1}(\gbf,1-\rho) -\eta_{1}(\hbf,1-\rho) +o(1)} } .
\end{align*}
Recall that $\rhoy \leq 0.5$ implies
\begin{align*}
    \eta_{1}(\gbf,\rho) +\eta_{1}(\hbf,\rho) \leq \eta_{1}(\gbf,1-\rho) -\eta_{1}(\hbf,1-\rho) ,
\end{align*}
and $\rhoy > 0.5$ implies 
\begin{align*}
    \eta_{1}(\gbf,\rho) +\eta_{1}(\hbf,\rho) \geq \eta_{1}(\gbf,1-\rho) +\eta_{1}(\hbf,1-\rho) .
\end{align*}
When $\rhoy \leq 0.5$, if $\eta_{1}(\gbf,\rho) +\eta_{1}(\hbf,\rho)<1$, then $\Prob  (E_{2} )  \rightarrow 1$. When $\rhoy \geq 0.5$, if $\eta_{1}(\gbf,1-\rho) +\eta_{1}(\hbf,1-\rho) <1$, then $\Prob  (E_{2} )  \rightarrow 1$ and the second part of Theorem~\ref{Theorem-SDP-oublv-censored} follows.
%%%%%%%%%%%%%%%%%%%%%%%%%%%%%%%%%%
\section{Proof of Theorem \ref{Theorem-SDP-tublv-censored}}
\label{Proof-Theorem-SDP-tublv-censored}

The proof is similar to the proof of Theorem~\ref{Theorem-SDP-tublv}. Here we just mention the proof outlines and important Lemmas for brevity.
The following Lemma declares the sufficient conditions for the optimum solution of~\eqref{tublv-sdp-censored} matching the true labels $x^*$.
\begin{Lemma}
\label{Lemma_sufficient_oublv-censored-t}
For the optimization problem~\eqref{tublv-sdp-censored}, consider the Lagrange multipliers 
\begin{equation*}
\lambda^{*}, \quad D^{*}=\mathrm{diag}(d_{i}^{*}), \quad
S^{*}.
\end{equation*}
If we have 
\begin{align*}
&S^{*} = D^{*}+\lambda^{*}\mathbf{J} -T A -\Ttwo  (A*A) ,\\
&S^{*} \succeq 0, \\
&\lambda_{2}(S^{*}) >  0 ,\\
&S^{*}x^{*} =0 ,
\end{align*}
then $(\lambda^{*}, D^*, S^*)$ is the dual optimal solution and $\Zsdp=x^{*}x^{*T}$ is the unique primal optimal solution of~\eqref{tublv-sdp-censored}.
\end{Lemma}
\begin{proof}
The proof is similar to the proof of Lemma~\ref{Lemma_sufficient_oublv}.
\end{proof}
Let
\begin{align*} 
d_{i}^{*} =&  T \sum_{j=1}^{n} A_{ij} x_{j}^{*}x_{i}^{*} +\Ttwo  \sum_{j=1}^{n} A_{ij}^2 x_{j}^{*}x_{i}^{*}.
\end{align*}
Then $D^{*}x^{*}=T A +\Ttwo  (A*A)$ and based on the definition of $S^{*}$ in  Lemma~\ref{Lemma_sufficient_oublv-censored-t}, $S^{*}$ satisfies the condition $S^{*}x^{*} =0$.
\begin{Lemma}
\label{Lemma_main_one-censored-t}
For $\delta = \frac{\log n}{ \log \log n}$,
\begin{align*}
    \Prob \Big( \min_{i\in [n]}~d_{i}^{*} \geq \delta  \Big ) \geq& 1-n^{1-\eta_{2}(\gbf,\rho) -\eta_{2}(\hbf,\rho) +o(1)} \\
    &-n^{1-\eta_{2}(\gbf,1-\rho) -\eta_{2}(\hbf,1-\rho) +o(1)} .
\end{align*}
\end{Lemma}
\begin{proof}
The proof is achieved by applying the Chernoff bound and taking the union bound. 
\end{proof}
Similar to the proof of Theorem~\ref{Theorem-SDP-tublv}, using Lemma~\ref{Lemma_main_one-censored-t}, it can be shown that
 $S^{*}\succeq 0$ and $\lambda_{2}(S^{*})>0$ with probability $1-o(1)$ if
\begin{align*}
\min \big \{ \eta_2(\gbf,\rhoy)+\eta_2(\hbf,\rhoy), \eta_2(\gbf,1-\rhoy)+\eta_2(\hbf,1-\rhoy) \big \} > 1.
\end{align*}

To prove the second part, we start to find when the maximum likelihood estimator fails. To this end, let 
\begin{align*}
    e(i,\Hset) \triangleq \sum_{j \in \Hset} T A_{ij} +\Ttwo A_{ij}^2 .
\end{align*}
The definition of events $F_1$, $F_2$, $E_1$, and $E_2$ in Theorem~\ref{Theorem-SDP-oublv} are used to show that with high probability $\Prob (F_{1} )  \rightarrow 1$ and $\Prob (F_{2} )  \rightarrow 1$. Also, the definitions for $C_{1}^{*}$, $C_{2}^{*}$, and $\Tset$ remain valid for this part.
We prove that $\Prob (F_{1} )  \rightarrow 1$, while $\Prob (F_{2} )  \rightarrow 1$ is proved similarly.
To show that $\Prob (F_{1} )  \rightarrow 1$, we must have
$\Prob (E_{1} )  \rightarrow 1$ and $\Prob (E_{2} )  \rightarrow 1$. 
It can be shown that $\Prob (E_{1} ) \geq 1-n^{-1+o(1)}$ without difficulty. 
\begin{Lemma}
\label{Lemma_converse_prob_1-censored-t}
Let $E_{2}' \triangleq \big \{ e(i,C_{1}^{*} \setminus \Tset)-e(i,C_{2}^{*}) \leq -\delta \big\}$.
Then
\begin{align*}
    \Prob  ( E_{2}'   ) \geq& n^{-\eta_{2}(\gbf,\rho) -\eta_{2}(\hbf,\rho) +o(1)} \\
    &+ n^{-\eta_{2}(\gbf,1-\rho) -\eta_{2}(\hbf,1-\rho) +o(1)}. 
\end{align*}
\end{Lemma}
\begin{proof}
The proof is achieved by applying Lemma~\ref{BASBM-N-Lemma 5} and the Chernoff bound. 
\end{proof}
Applying Lemma~\ref{Lemma_converse_prob_1-censored-t} yields
\begin{align*}
\Prob  (E_{2} ) &=1-\prod_{i\in \Tset}   [ 1-\Prob  ( E_{2}'   )   ] \\
&\geq 1-e^{ -n^{-\eta_{2}(\gbf,\rho) -\eta_{2}(\hbf,\rho) +o(1)} -n^{-\eta_{2}(\gbf,1-\rho) -\eta_{2}(\hbf,1-\rho) +o(1)} } .
\end{align*}
If
\begin{align*}
    \min \big \{\eta_2(\gbf,\rhoy)+\eta_2(\hbf,\rhoy), \eta_2(\gbf,1-\rhoy)+\eta_2(\hbf,1-\rhoy) \big \} < 1 ,
\end{align*}
then $\Prob  (E_{2} )  \rightarrow 1$ and the second part of Theorem~\ref{Theorem-SDP-tublv-censored} follows.
%%%%%%%%%%%%%%%%%%%%%%%%%%%%%%%%%%%%%%%%%%%%%%%%%%%%%%%%%%%%%
%References are important to the reader; therefore, each citation must be complete and correct. If at all possible, references should be commonly available publications.

%\ifCLASSOPTIONcaptionsoff
%\newpage
%\fi 

\bibliographystyle{IEEEtran}
\bibliography{References}

\end{document}